\documentclass[11pt]{article}

\usepackage{graphicx,amsmath,amsfonts}
\usepackage{algorithm}
\usepackage{algorithmic}
\usepackage{caption}
\usepackage[usenames,dvipsnames]{color}
\usepackage{comment}
\usepackage[round]{natbib}
\usepackage[height=9in,a4paper,hmargin={3cm,3cm}]{geometry}
\usepackage{graphicx}
\usepackage{caption}
\usepackage{subcaption}
\newcommand{\NE}{network effect}
\usepackage{amsmath}

\usepackage{amssymb}

\usepackage{hyperref}
\definecolor{darkred}{RGB}{150,0,0}
\definecolor{darkgreen}{RGB}{0,150,0}
\definecolor{darkblue}{RGB}{0,0,200}
\hypersetup{colorlinks=true, linkcolor=darkred, citecolor=darkgreen, urlcolor=darkblue}

\numberwithin{equation}{section}

\def \endprf{\hfill {\vrule height6pt width6pt depth0pt}\medskip}
\newenvironment{proof}{\noindent {\bf Proof} }{\endprf\par}

\newtheorem{theorem}{\textbf{Theorem}}
\newtheorem{lemma}{\textbf{Lemma}}
\newtheorem{corollary}{\textbf{Corollary}}
\newtheorem{example}{\textbf{Example}}
\newtheorem{definition}{\textbf{Definition}}

\newtheorem{proposition}{\textbf{Proposition}}
\newtheorem{assumption}{\textbf{Assumption}}

\begin{document}

\title{Strategic Dynamic Pricing with Network Effects}
\author{Ali Makhdoumi\thanks{%
Department of Electrical Engineering and Computer Science, Massachusetts
Institute of Technology, Cambridge, MA, 02139 \texttt{makhdoum@mit.edu}}
\and Azarakhsh Malekian\thanks{%
Rotman School of Management, University of Toronto, Toronto, ON, M5S 3E6
\texttt{azarakhsh.malekian@rotman.utoronto.ca}} \and Asuman Ozdaglar\thanks{%
Department of Electrical Engineering and Computer Science, Massachusetts
Institute of Technology, Cambridge, MA, 02139 \texttt{asuman@mit.edu}}}
\date{\today}
\date{}
\maketitle
\begin{abstract}We study the optimal pricing strategy of a monopolist selling homogeneous goods to $n$ buyers over multiple periods. The customers choose their time of purchase to maximize their payoff that depends on their valuation of the product, the purchase price, and the utility they derive from past purchases of others, termed the network effect. We first show that the optimal price sequence is non-decreasing. Therefore, by postponing purchase to future rounds, customers trade-off a higher utility from the network effects with a higher price. 
 We then show that a customer's equilibrium strategy can be characterized by a threshold rule in which at each round a customer purchases the product if and only if her valuation exceeds a certain threshold. This implies that customers face an inference problem regarding the valuations of others, i.e., observing that a customer has not yet purchased the product, signals that her valuation is below a threshold. We consider a block model of network interactions, where there are blocks of buyers (with size equal to constant fraction of total number of buyers) subject to the same network effect. A natural benchmark, this model allows us to provide an explicit characterization of the optimal price sequence asymptotically as the number of agents goes to infinity, which notably is linearly increasing in time with a slope that depends on the network effect through a scalar given by the sum of entries of the inverse of the network weight matrix. 
  Our characterization shows that increasing the ``imbalance'' in the network defined as the difference between the in and out degree of the nodes increases the revenue of the monopolist. We further study the effects of price discrimination and show that in earlier periods monopolist offers lower prices to blocks with higher Bonacich centrality to encourage them to purchase, which in turn further incentivizes other customers to buy in subsequent periods.
\end{abstract}

\maketitle



\section{Introduction}
The benefits that users derive from various products such as digital products (e.g., computer software and smartphone apps) and electronics (e.g., smartphones, hardware devices, and computers) depend, among other things, on the other users who have purchased the product before and thus have contributed to improvements of various aspects of the product. User's purchase decisions in settings with such externalities, termed network effects, will focus on not just whether but when to make a purchase. This implies that a seller will choose a dynamic price path for the product that aims to build the central early user population to increase the network effects and the purchase possibility of the product in later periods. Despite the ubiquity of these issues, there is little work on dynamic pricing with network effects.  

In this paper, we study the problem of dynamic pricing in the presence of network effects. We consider a dynamic game between a monopolist seller and a set of buyers. All buyers and seller are forward-looking. The seller announces and \emph{commits} to a price sequence and buyers decide whether and when to buy a single item. The utility of each buyer depends on her valuation of the product, the price sequence, and the (weighted) number of other customers who have already bought the product. 
We first show that in this setting the optimal price sequence is non-decreasing. We then show that the equilibrium purchase decision of buyers (in Perfect Bayesian Equilibrium) can be characterized by a threshold rule in which buyers purchase at different rounds if and only if their valuations exceed a certain threshold. This characterization implies that customers face a learning problem regarding the valuations of others since if a round has reached and a customer has not yet purchased the item, her valuation must be below the threshold. Therefore, at any round the belief of a customer regarding the valuations of the remaining customers gets updated using Bayes' rule. As the optimal price sequence is non-decreasing, by postponing purchase to future rounds each buyer faces the following trade-off: on one hand, she has to pay a higher price, and on the other hand, her utility from the network effects becomes larger.

Building on this characterization, we find the optimal pricing strategy in a block model setting. In particular, we consider a block model with $m$ blocks such that each block $h$ has a constant fraction of the total number of buyers. The networks effects are captured by a matrix $E \in \mathbb{R}^{m \times m}$, where $E_{hh'}$ denotes the network effect of a buyer in block $h'$ on a buyer in block $h$. 
This model provides a natural benchmark in which there are blocks of users subject to the same network effect while still allowing diverse interactions among these blocks. Different blocks may for example represent communities with different characteristics (see \citet[Chapter 7]{tirole1988theory} and \citet[Chapter 8]{talluri2006theory}). Most importantly, this model allows us to write the seller's expected revenue as a multivariate Bernstein polynomial which enables us to use asymptotic convergence theory of these polynomials and explicitly characterize the optimal price sequence (see \cite{lorentz2012bernstein}).

Interestingly, for any distribution of buyers' valuations (under some regularity conditions) we find the closed-form solution of the optimal price sequence asymptotically as the number of users goes to infinity. The optimal price sequence is linearly increasing, and our characterization shows that the properties of both optimal price sequence and the optimal revenue depend on the quantity $1/\left(\mathbf{1}^{T} E^{-1} \mathbf{1}\right)$ which we term the \emph{\NE}. In particular, the extent of the price difference at two consecutive rounds (slope of the price path) is greater for higher \NE~and the optimal revenue is increasing and convex in \NE. The \NE~(and hence the optimal revenue) is higher for ``imbalance'' networks.\footnote{A directed network is called balance if for each node the out-degree and in-degree are equal.} More precisely, for a given sum of network effects (i.e., sum of the entries of $E$), by decreasing sum of the multiplication of out-degree and in-degree of blocks the revenue increases. 

Moreover, we consider a setting with price discrimination, where monopolist offers different prices to different blocks, and characterize the optimal price sequence. We establish that the optimal price sequence is linearly increasing with a slope which is in the form of a ``weighted Bonacich centrality''. Our results indicate that in earlier periods monopolist offers lower prices to blocks with higher centrality to encourage them to purchase, which in turn further incentivizes other customers to buy at higher prices in the subsequent periods. We also consider a variation of our model with utility from all purchases and characterize the optimal price sequence and revenue. 


\subsection{Related Literature}
Our paper relates to two sets of works: (i) the study of markets with network effects and (ii) the study of markets with forward-looking strategic buyers.
\subsubsection{Network Externalities}
Markets for products with network effects has been first studied in \cite{rohlfs1974theory}, \cite{katz1985network}, and \cite{farrell1985standardization}.
Given the importance of network effects in markets, empirical investigations have examined its implications in a variety of industries. In particular, \cite{au2001should} examine the adoption of electronic bill and payment technology and show the existence of network effects and its implications, \cite{gallaugher2002understanding} empirically study the market for Web server software and establish the presence of network effects, and \cite{brynjolfsson1996network} empirically study the network effects in software product market and confirm that the network effects significantly has increased the price of products.

Moreover, on the theory side, a line of research has examined the strategic and welfare implications of network effects. In particular, \cite{candogan2012optimal}, \cite{cohen2013designing}, and \cite{bloch2013pricing} study the optimal static price sequence of a seller selling a divisible good (service) to consumers with network effects. They consider a two-stage game in which a seller decides on the prices, and then buyers decide their consumption in an equilibrium. Given a set of prices, their model takes the form of a network game among agents that interact locally, which relates to a series of papers such as \cite{ballester2006s, bramoulle2007public, galeotti2009influencing}, and \cite{bramoulle2014strategic}. Related models are more recently investigated in \cite{fainmesser2016pricing}, \cite{Manshadi17}, and \cite{belloni2017mechanism}. In particular, \cite{Manshadi17} consider promotion planning of network products and study the effect of network structure.


\subsubsection{Strategic Buyers}
A number of papers in the literature consider strategic forward-looking buyers who make inter-temporal purchasing decisions with the goal of maximizing their utility. Many empirical works confirm that assuming myopic customer behavior is no longer a tenable assumption (see e.g. \cite{li2014consumers}). The importance of forward-looking customer behavior in shaping firms' pricing decision has been broadly identified by practitioners and a recent literature has theoretically studied its implications (see \cite{liu2008strategic, horner2011managing}, \cite{board2016revenue, pai2013optimal, borgs2014optimal, cachon2015price, chen2015robust, yang2015dynamic, lobel2015optimizing, dilme2016revenue, bernstein2016dynamic}). In particular, \cite{besbes2015intertemporal} study the optimal price sequence of a committed seller that faces customers arriving over time with heterogeneous willingness to wait before making a purchase. They show that cyclic pricing policies are optimal for this setting. \cite{ajorlou2016dynamic} consider a setting in which customers know about the existence of a product through the word-of-mouth communication and study the pricing strategy of the seller. \cite{papanastasiou2016dynamic} consider a setting in which forward-looking customers learn the quality (unknown and fixed) of a product from the reviews of their peers and study the pricing strategy of the seller. \cite{lingenbrink2018signaling} consider a setting with strategic customers and uncertain inventory and find the revenue-optimal signaling mechanism (i.e., the signaling that ``persuades'' customers to purchase at higher price).\footnote{Our paper also relates to the literature on the 
``Coase conjecture'' \cite{coase1972durability, gul1986foundations}.  
 ``Coasian dynamics'' (\cite{hart1988contract}) consist of two properties: (i) higher valuation buyers make their purchase no later than lower valuation buyers (skimming property) and (ii) equilibrium price sequence is non-increasing over time (price monotonicity property). In this paper, we show that the second property does not hold when network effects are present.}
\subsection{Notation}
For any matrix $M \in \mathbb{R}^{m \times n}$, we use both $[M]_{ij}$ and $M_{ij}$ to denote the entry at $i$th row and $j$th column. We show vectors with bold face letters. For a matrix $M$, $\rho(M)$ and $||M||_{\infty}$ denote the spectral radius and infinity norm of $M$, defined as $\rho(M) =\max\{|\lambda|~:~ \lambda \text{ is eigen value of } M\}$, and $||M||_{\infty} = \max_{1 \le i\le m} \sum_{j=1}^n |M_{ij}|$, respectively. The vector of all ones is shown by $\mathbf{1}$, where the dimension of vector is clear from the context. For any event $\mathcal{E}$, $\mathbf{1}\{\mathcal{E}\} =1$ if $\mathcal{E}$ holds and $0$, otherwise. For any integer $n$, we let $[n]=\{1, \dots, n\}$. We denote the transpose of vector $\mathbf{x}$ and matrix $M$ by $\mathbf{x}^T$ and $M^T$, respectively. For two vectors $\mathbf{x}, \mathbf{y} \in \mathbb{R}^m$, $\mathbf{x} \ge \mathbf{y}$ means entry-wise inequality, i.e. $x_i \ge y_i$, for all $i \in [m]$. For two matrices $A, B \in \mathbb{R}^{m \times m}$, $A \ge B$ means entry-wise inequality, i.e. $A_{ij} \ge B_{ij}$, for all $i,j \in [m]$. We show a weighted directed network by $(V, G)$ where $V=\{1, \dots, n\}$ represents the set of nodes and $G_{ij}$ represents the weight of the edge from $j$ to $i$. In-degree and out-degree of node $i$ are denoted by $d_i^{(\text{in})}= \sum_{j=1}^n G_{ij}$ and $d_i^{(\text{out})}= \sum_{j=1}^m G_{ji}$. A network is called \emph{symmetric} iff $G_{ij}=G_{ji}$ for all $i, j \in [n]$ and is called \emph{balance} iff $d_i^{\text{(out)}}=d_i^{\text{(in)}}$ for all $i \in [n]$. We let $\mathcal{P}_{[0,1]}(\cdot)$ denote the projection operator onto the interval $[0,1]$.
\subsection{Outline}
The rest of the paper is organized as follows. In Sections \ref{sec:Model}, \ref{sec:Prelim}, and \ref{sec:IslandModel} we describe the model and provide preliminary characterizations of buyers and seller strategy. In Section \ref{sec:OptimalPricing} we characterize the optimal price sequence. In Section \ref{sec:networkeffects}, we study the effects of network structure on the optimal price sequence and optimal revenue. Finally, in Section \ref{sec:PriceDiscrimination} we characterize optimal price sequence with price discrimination, leading to concluding remarks in Section \ref{sec:Conclusion}. All of the omitted proofs as well as a variation of our model and analysis are presented in the Appendix. 
\section{Model Description}\label{sec:Model}
We consider a dynamic game between a monopolist with infinitely many homogeneous items and $n$ buyers in $T$ \emph{rounds}. For our analysis, we find it more convenient to index the rounds in decreasing order, i.e., round $t$ refers to \emph{period} $T+1-t$ (there are $t-1$ remaining periods until the end of selling horizon). The monopolist announces a price sequence $\mathbf{p}=(p_{T}, p_{T-1}, \dots, p_1)$ where $p_t$ is the price offered for the item at round $t$. At each round, buyers decide whether to buy the item or postpone it to future rounds. We let $h^t$ denote the set of buyers that buy the item at round $t$ at price $p_t$. The history available to buyers at round $t$ is given by $H^t=\{h^{T+1}, h^{T}, \dots, h^{t+1}\}$ (we also let $h^{T+1} = \emptyset$). Each buyer has a valuation $v$ in $[0,1]$  drawn independently from a known continuous distribution $F_v(\cdot)$. The utility of each buyer depends on her valuation, the price sequence, and a network effect term. More specifically, we assume that buyer interactions are captured by a directed weighted graph $(V, G)$, where $V=\{1, \dots, n\}$ is the set of buyers and $G=[g_{ij}]_{i,j \in V}$ is a weight matrix with $g_{ij} \ge 0$ representing the utility gain $i$ derives from $j$'s purchase (note that $g_{ij}$ can be different from $g_{ji}$). Thus, the utility of buyer $i$ with valuation $v^{(i)}$, if she buys the item at round $t$ with price $p_t$ is given by
\begin{align}\label{eq:UtilityofBuyers}
u_i\left(\mathbf{v}, \mathbf{p}, H^t \right)= v^{(i)} - p_t +  \sum_{j \neq i} g_{ij} \mathbf{1}\{ j \text{ has bought at round } s>t \},
\end{align}
where $\mathbf{v}=(v^{(1)}, \dots, v^{(n)})$. The third term is the weighted sum of other buyers who have bought the item before round $t$ and represents the ``network effect'' on buyer $i$'s utility. This utility model captures situations in which other buyers purchase decisions affect a buyer's utility by improving the product through their use. Hence, at the time a buyer makes her purchase decision and utilizes the product, what matters is past users that determine how improved the product is. Note that consistent with this interpretation and full rationality (in perfect Bayesian equilibrium), even though buyers derive utility from purchases in the past, they are \textit{forward-looking}, i.e., they take into account all future behavior and decide to postpone the purchase if it is optimal. In Appendix \ref{sec:future}, we consider a variation of our model with utility from all purchases (i.e., not only the previous purchases). For this alternative model, we find the optimal price sequence and the optimal revenue of the monopolist. 


We next describe the buyers' strategies. A given price sequence $\mathbf{p}=(p_T, \dots, p_1)$ induces a dynamic incomplete information game among buyers. A (pure) strategy for buyer $i$ is a sequence $\{b_i^t\}_{t=1}^T$, where $b_i^t$ is a mapping from $\mathbb{R} \times H^t \times \mathbb{R}^T$ into $\{0, 1\}$, mapping buyer $i$'s valuation, the history of the game, and the price sequence into a purchase decision. A perfect Bayesian equilibrium is a collection of strategies $\{b_i^t\}_{t=1}^T$, for $i \in [n]$ such that buyer $i$ maximizes her expected utility given her belief (updated in a Bayesian manner) and strategies of other buyers. In particular, in a perfect Bayesian equilibrium a buyer $i \not\in H^t$ buys the item at round $t$ with price $p_t$ if and only if 
\begin{align*}
u_i\left(\mathbf{v}, \mathbf{p}, H^t \right) \ge \max_{s<t}\mathbb{E}\left[ u_i\left(\mathbf{v}, \mathbf{p}, H^s \right) \mid H^t\right], 
\end{align*}
where the left hand side is the utility of buyer $i$ if she purchases at round $t$ and the right hand side is the maximum of expected utility from purchasing in any of the future rounds (i.e., $s<t$). The expectation is taken over the belief of buyer $i$ regarding the other buyers' valuations. 

For a given price sequence $\mathbf{p}$, the expected revenue of the monopolist is\footnote{We use the terms monopolist and seller interchangeably. We also use the terms user, customer, and buyer interchangeably.}
\begin{align}\label{eq:Sellerutility}
\sum_{t=1}^T \sum_{i=1}^n p_t \mathbb{E} \left[ \mathbf{1}\{ i \text{ buys at round } t \text{ with price } p_t\} \right],
\end{align}
where we normalized the marginal cost of the monopolist to zero. We refer to the price sequence that maximizes Eq. \eqref{eq:Sellerutility} as the \textit{optimal price sequence} and the corresponding revenue as the \textit{optimal revenue} (we also use the term \textit{optimal normalized revenue} which is equal to the optimal revenue divided by the number of buyers). The ability of the seller to commit to a price sequence is important for our results. Without such commitment, a Coase conjecture-type reasoning would create a downward pressure on prices and would tend to reduce seller's revenue (\cite{coase1972durability}). Though such commitment is not possible in some settings, many sellers can build a reputation for such commitment, for instance, by creating explicit early discounts which will be lifted later on (see \citet[Chapter 8]{talluri2006theory} for a discussion of committed pricing).

\section{Preliminary Characterizations}\label{sec:Prelim}
In this section, we show the optimal price sequence is non-decreasing and characterize the purchase decision of customers in equilibrium. Each buyer faces an optimal stopping problem, choosing the round (if any) along a sequence of prices at which to accept the offered price and exit the game given the strategy of other buyers. We next show that the optimal price sequence is non-decreasing and buyer $i$ chooses to purchase the item only if her valuation exceeds a time and history dependent threshold denoted by $v^{(i)}_t (H^t)$ that satisfies $v^{(i)}_{t}(H^{t}) \ge v^{(i)}_{t-1} (H^{t-1})$ for all $H^{t-1}$. In the rest of the paper, we use $v^{(i)}_t (H^t)$ and $v^{(i)}_t$ interchangeably and refer to them as \emph{critical thresholds}.
\begin{proposition}\label{Pro:SingleCrrossing}
\begin{itemize}
\item [(a)] The seller's optimal price sequence is non-decreasing, i.e., any optimal price sequence has a corresponding non-decreasing price sequence in which the equilibrium path (the purchase decision of buyers) remains the same. 
\item [(b)] Given a non-decreasing price sequence, the purchase decision of buyer $i \in [n]$ in any equilibrium is a thresholding decision, i.e., for each $i$ there exists a sequence $\{v^{(i)}_t\}_{t=T}^{1}$, such that $v_t^{(i)} \ge v_{t-1}^{(i)}$ for $2 \le t\le T $ and buyer $i$ purchases at round $t$ if and only if $v^{(i)} \ge v^{(i)}_t$. 
\end{itemize}
\end{proposition}
Proposition \ref{Pro:SingleCrrossing} is a crucial observation on which much of the rest of our analysis builds. Technically, it is simple and relates to previous results in dynamic settings with preferences satisfying single crossing. Its implications in our model are far-reaching, however. Without network effects, increasing (non-decreasing) price sequence would be impossible to sustain. Because with increasing prices, high-valuation buyers will tend to be the first ones to purchase, and if lower-valuation buyers prefer not to purchase early on (with low prices), they would also prefer not to purchase later on with higher prices (as there is no benefit from network effects). This phenomenon is transformed in the presence of network effects. Now because the increase in the number of users over time raises the network effect term (regardless of the exact form of network interactions), lower-valuation buyers might prefer to buy later and at higher prices. In fact, it is not optimal for the seller to have a strictly decreasing price sequence, because this would induce all buyers to delay, while an increasing price sequence would induce high-valuation buyers to purchase early, while lower-valuation ones wait and purchase once the network effect is higher.

In the next lemma we provide a relation for the critical thresholds describing the buyers' purchase decision. 
\begin{lemma}\label{Lem:buyer'sbehavior}
Given a non-decreasing price sequence $\mathbf{p}$, for $t=T, \dots, 1$ the critical thresholds satisfy the following indifference condition
\begin{align}\label{eq:indifferenceequation}
 \sum_{j \in [n] \setminus \left(H^t \cup \{i\} \right)} g_{ij} \mathbb{P}\left[v^{(j)} \ge v^{(j)}_{t} \mid v^{(j)} \le v^{(j)}_{t+1}\right]& =  p_{t-1}- p_t , \qquad i \not\in H^t,
\end{align}
with $v_{T+1}^{(i)}=1$ for all $i \in [n]$. Moreover, if $v^{(j)}_{t}$'s are in $[0,1]$, then we have 
\begin{align*}
\mathbb{P}\left[v^{(j)} \ge v^{(j)}_{t} \mid v^{(j)} \le v^{(j)}_{t+1}\right]= 1- \frac{F_v \left(  v^{(j)}_t  \right)}{F_v \left(  v^{(j)}_{t+1}  \right)} .
\end{align*}
\end{lemma}
This relation is defined by noting that a buyer $i$ with valuation equal to $v_t^{(i)}$ must be indifferent between buying at round $t$ and waiting until round $t-1$. 

The sequence $\{v^{(i)}_t\}_{t=T+1}^{1}$, $i \in [n]$ depends on the history of the game and Lemma \ref{Lem:buyer'sbehavior} provides an indifference condition for it, with boundary conditions $v^{(i)}_{T+1} = 1$, $i \in [n]$. This is because in the first period, the seller and buyers have not yet learned anything about buyer $i$'s valuation. Note that each buyer faces an inference problem regarding the valuation of the other buyers. In particular, if round $t$ with history $H^t$ is reached and a buyer $j$ has not purchased the item, then all buyers know that $v^{(j)}$ belongs to $[0, v^{(j)}_t)$ with cumulative density function 
\begin{align*}
\mathbb{P}\left[v^{(j)}\le x \mid H^t \right] = \frac{F_v(x)}{F_v(v^{(j)}_t)}, \qquad \forall x \in \left[ 0, v^{(j)}_t \right),
\end{align*}
which is obtained via Bayes' rule.

In the next example, we illustrate two main challenges in analyzing the equilibrium of this game. First, we illustrate that finding the optimal price sequence is a complicated task because the monopolist needs to take into account all possible histories (i.e., $(T+1)^n$ histories for a game with $n$ buyers in $T$ rounds). Second, we illustrate the possibility of having multiple equilibria. 
\begin{figure}
 \begin{center}
    \includegraphics[width=0.24\textwidth]{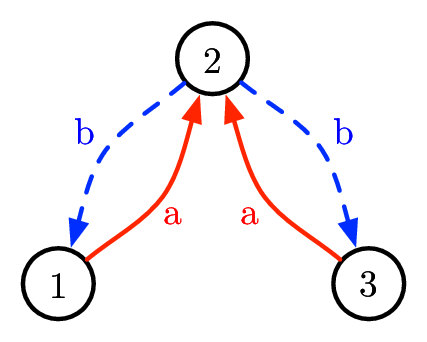}
      \end{center}
\caption{Illustration of the graph capturing the network effects used in Example \ref{Example:multiple} for $a=\frac{4}{5}$ and $b=\frac{3}{5}$.}
    \label{fig:Example}

\end{figure}
\begin{example}\label{Example:multiple}
\textup{We consider a game with $3$ buyers in $2$ periods ($n=3$ and $T=2$) where the network effects are represented by 
\begin{align*}
G= \begin{small}\begin{pmatrix} 
 0 &&& a &&& 0 \\
b &&& 0 &&& b \\
0 &&& a &&& 0 
\end{pmatrix} \end{small},
\end{align*}
 for $a=4/5$ and $b=3/5$ (see Figure \ref{fig:Example}). We also let the valuations to be uniform, i.e., $F_v(x)=x$, for all $x \in [0,1]$. Given a non-decreasing price sequence $(p_2, p_1)$, if the critical thresholds are in $(0,1)$, using Lemma \ref{Lem:buyer'sbehavior} for the first period (i.e., $t=2$) leads to 
\begin{align*}
v_2^{(1)}-p_2 & = v_2^{(1)}-p_1+ a (1- v_2^{(2)}), \quad v_2^{(3)}-p_2  = v_2^{(3)}-p_1+ a (1- v_2^{(2)}), \\
v_2^{(2)}-p_2 & = v_2^{(2)}-p_1+ b \left( \left(1- v_2^{(1)}\right) + \left(1- v_2^{(3)}\right) \right),
\end{align*}
where we used $v_3^{(i)}=1$, $i=1,2,3$. This set of equations have multiple solutions, leading to multiple strategies for the first period of the game. In particular, the critical thresholds satisfy
\begin{align}\label{eq:Example:FirstRoundThresholds}
v_2^{(2)} = 1- \frac{p_1-p_2}{a}, \quad v_2^{(2)} + v_2^{(3)} = 2- \frac{p_1-p_2}{a}.
\end{align}
Using Lemma \ref{Lem:buyer'sbehavior} for the second period (i.e., $t=1$), we have 
\begin{align}\label{eq:Example:SecondRoundThresholds}
v_1^{(1)} & =  p_1- a \mathbf{1}\{2 \text{ bought at price } p_2\} , \quad v_1^{(3)}  =  p_1- a \mathbf{1}\{2 \text{ bought at price } p_2\}, \nonumber \\
v_1^{(2)} & =  p_1- b \left( \mathbf{1}\{1 \text{ bought at price } p_2\} + \mathbf{1}\{3 \text{ bought at price } p_2\} \right). 
\end{align}
Given any set of critical thresholds, the expected revenue becomes
\begin{align}\label{eq:Example:Revenue}
 \sum_{S \subseteq \{1, 2, 3\}}\left( \prod_{i \in S} \mathbb{P}\left[ v^{(i)} \ge v_2^{(i)}\right] \right)  \left( p_2 |S| + p_1\left( \sum_{i \in \{1, 2, 3\} \setminus S }\mathbb{P}\left[ v^{(i)} \ge v_1^{(i)} \mid H^2=S\right] \right) \right),
\end{align}
where the expectation is taken over all histories for $H^2$ ($8$ possible histories). The first term $p_2 |S|$ is the revenue in the first period and the second term is the expected revenue in the second period, given $H^2$.
%
 Maximizing the expected revenue subject to Eqs. \eqref{eq:Example:FirstRoundThresholds} and \eqref{eq:Example:SecondRoundThresholds}, leads to the following pairs (among many others) of price sequence and buyers' strategies:
\begin{itemize}
\item Symmetric equilibrium: in this equilibrium the buyers' decision depend on their valuation, price, and the network effects, and not on their identity. Therefore, the strategies of buyers $1$ and $3$ are the same, i.e., $v_1^{(1)}=v_1^{(3)}$ and $v_2^{(1)}=v_2^{(3)}$ which leads to price sequence $p_1=.6$ and $p_2=.48$ with expected revenue $.38$. The corresponding buyers' strategies in the first period are determined by critical thresholds $v_2^{(2)}= 1- (p_1-p_2)/a$, $v_2^{(1)}=v_2^{(3)}= 1- (p_1-p_2)/2b$, and in the second period by critical thresholds given in Eq. \eqref{eq:Example:SecondRoundThresholds}. Note that all these critical thresholds are in $(0,1)$. 
\item Asymmetric equilibrium: suppose the buyers' strategies in the first period are determined by critical thresholds $v_2^{(2)}= 1- (p_1-p_2)/a$, $v_2^{(1)}= 1- (p_1-p_2)/3b$, $v_2^{(3)}= 1- 2(p_1-p_2)/3b$. The optimal price sequence becomes $p_1=.6$ and $p_2=.42$ with expected revenue $.52$. The buyers' strategies in the second period are determined by critical thresholds given in Eq. \eqref{eq:Example:SecondRoundThresholds}. Note that again all these critical thresholds are in $(0,1)$. 
\end{itemize}
}
\end{example}
To overcome the challenges illustrated in Example \ref{Example:multiple} and to obtain a tractable expected revenue, we will consider a block model as described in the next section and focus on the symmetric equilibrium concept. Intuitively, with this setting the sample equilibrium path is close to its expectation which enables us to use techniques from probability theory (namely, Bernstein polynomial convergence; see \cite{lorentz2012bernstein}) to find a closed-form characterization for the expected revenue as well as the optimal price sequence. 
\section{The Block Model and Buyers' Equilibrium}\label{sec:IslandModel}
We consider a ``block model'' (first introduced in \cite{white1976social, holland1983stochastic}) in which the buyers are partitioned into $m$ blocks. Block $k \in [m]$ denoted by $G_k$ has $n_k=\alpha_k n$ many buyers for some $\alpha_k \in (0,1]$. We let $A \in \mathbb{R}^{m \times m}$ be a diagonal matrix with $A_{ii}=\alpha_i$. The utility gains of buyers in block $h$ from purchase decision of buyers in block $k$ are equal and denoted by $E_{hk}$. Hence we can capture all the gains by a \emph{network matrix} $E \in \mathbb{R}^{m \times m}$. Formally, we let 
\begin{align*}
g_{ij}& = \frac{E_{kh}}{n}, \quad i\in G_k, j \in G_{h}, k, h \in [m], k\neq h, \\
g_{ij}&= \frac{E_{kk}}{n}, \quad i,j \in G_k, k \in [m],
\end{align*}
where the normalization by $n$ is to guarantee that the network effects term in buyers' utilities does not grow with $n$ and is comparable to valuations which are in $[0,1]$. For instance, if $l_h \in [0, \alpha_h n]$ many buyers from block $h$ purchase the product, then the network effects of a buyer in $G_k$ is $\sum_{h \in [m]} E_{kh} \frac{l_h}{n}$. This model provides a natural benchmark in which there are blocks of users subject to the same network effect while still allowing diverse interactions among these blocks. These blocks may for example represent communities with dense linkages within themselves. Most importantly, this model allows us to write the monopolist's expected revenue as a multivariate Bernstein polynomial which enables us to use convergence results of these polynomials and explicitly characterize the optimal price sequence.

The following provides the definition and convergence of multivariate Bernstein polynomial (see \citet[Chapter 2.9]{lorentz2012bernstein}).
\begin{definition}[Multivariate Bernstein Polynomial]
\textup{
For any function $f: \Delta_k \to \mathbb{R}$, where $\Delta_k= \{\mathbf{x} \in \mathbb{R}^k: ~x_i \ge 0, i \in [k], ~\sum_{i=1}^k x_i\le 1\}$ is $k$-dimensional simplex, multivariate Bernstein polynomial is defined as
\begin{align*}
B_{f,n}(x_1, \dots, x_k) = \sum_{\substack{r_i \ge 0, i \in [k]\\ \sum_{i=1}^k r_i \le n}} f\left(\frac{r_1}{n}, \dots, \frac{r_k}{n} \right) \binom{n}{r_1, \dots, r_k} x_1^{r_1} \cdots x_k^{r_k}  \left(1- x_1 - \dots - x_k \right)^{n-r_1- \dots -r_k},
\end{align*}
where 
\begin{align*}
 \binom{n}{r_1, \dots, r_k} = \frac{n!}{r_1! \dots r_k! (n- r_1 - \dots - r_k)!}. 
\end{align*}
}
\end{definition}
\begin{theorem}[\cite{lorentz2012bernstein}] \label{Thm:BersteinSimplex}
If $f:\Delta_k \to \mathbb{R}$ is continuous, then we have $B_{f,n} \to f$ uniformly. 
\end{theorem}
As shown in Example \ref{Example:multiple} there exist multiple equilibria for buyers' purchase decisions; however, there exists a unique symmetric equilibrium in which a buyer's strategy depends on her valuation and her network effects, not on her identity. In the rest of the paper, we use the symmetric equilibrium concept and refer to it as \emph{buyers' equilibrium}.\footnote{Symmetric equilibrium is used as a selection device among multiple equilibria which is widely used in the literature. See \cite{gul1986foundations}, \cite{chen2012name}, \cite{horner2011managing} for dynamic pricing settings, \citet[Chapter 4]{krishna2009auction} for auction setting, and \cite[Chapter 8]{talluri2006theory} for pricing games.} Using Lemma \ref{Lem:buyer'sbehavior}, the buyers' equilibrium is characterized by critical thresholds defined for blocks described in the next corollary. 
\begin{corollary}[Buyers' Equilibrium]\label{Cor:SymmetricEquil}
\textup{ Given a non-decreasing price sequence $\mathbf{p}$, for $t=T, \dots, 1$ the critical thresholds satisfy the following indifference condition
\begin{align}\label{eq:Def:SymmetricEquil}
 \sum_{k \in [m]} \left.| G_k \setminus H^t \right.| ~ \frac{E_{hk}}{n}  ~ \mathbb{P}\left[v \ge v^{(k)}_{t} \mid v \le v^{(k)}_{t+1}\right]=  p_{t-1}- p_t , \qquad h \in [m],
\end{align}
with $v_{T+1}^{(h)}=1$ for all $h \in [m]$. Moreover, in the buyers' equilibrium any remaining buyer in block $G_h$ purchases at price $p_t$ if and only if her valuation exceeds $v^{(h)}_t$.
}
\end{corollary} 

\section{Optimal Price Sequence}\label{sec:OptimalPricing}
In this section, we characterize the optimal price sequence for any valuation distribution $F_v(\cdot)$ and network effects $E$ under the following regularity conditions.
\begin{assumption}\label{assump:generalnetworkNonUniformValuations}
\textup{ Matrix $E$ is invertible and the distribution $F_v(\cdot)$ is such that 
\begin{align*}
x- \frac{1- F_v(x)}{f_v(x)} - \frac{F_v(x)}{\mathbf{1}^T E^{-1} \mathbf{1}}
\end{align*}
is non-decreasing. Also, the matrix $E$ is such that $\mathbf{0} \le \left(E A \right)^{-1} \mathbf{1} \le \left( \mathbf{1}^T E^{-1} \mathbf{1} \right) \mathbf{1}$.
}
\end{assumption}
The assumption on $F_v(\cdot)$ is the analogous of the regularity condition (i.e., $x- \frac{1-F_v(x)}{f_v(x)}$ is non-decreasing) and is used to guarantee the uniqueness of the optimal price sequence. Indeed, without network effect (i.e., with $E=0$) this assumption reduces to the regularity condition which is used in optimal auction design (see \cite{myerson1981optimal}). The assumptions on the network matrix $E$ guarantee that the critical thresholds are interior (i.e., in $(0,1)$) and therefore enable their explicit characterization. As an example, we next show that for uniform valuations and a \emph{weakly-tied block model}, i.e., $E= I + \delta C$, for small $\delta$, Assumption \ref{assump:generalnetworkNonUniformValuations} holds. This network matrix captures situations in which inter block interactions are weak. 

\begin{lemma}\label{Lem:DeltaSmall}
Suppose $E= I + \delta C$ with $C \ge 0$. If $\delta <  \frac{1}{2 \mathbf{1}^T C \mathbf{1}}$, then Assumption \ref{assump:generalnetworkNonUniformValuations} holds for uniform valuations and $\alpha_i=\frac{1}{m}$, for $i \in [m]$.  
\end{lemma}

Our key result presented next provides an explicit characterization of optimal prices and optimal revenue as a function of the network effects.  

\begin{theorem}\label{thm:NonUniformValuations}
Suppose Assumption \ref{assump:generalnetworkNonUniformValuations} holds. The optimal price sequence in the limit (as $n \to \infty$) is given by 
\begin{align}\label{eq:NonUniformValuationsprices}
p_t & = (T-t) \frac{1-F_v(p_T)}{T \left( \mathbf{1}^T E^{-1} \mathbf{1} \right)} + p_T, \quad t=1, \dots, T,
\end{align}
where $p_T$ is the solution of 
\begin{align*}
p_T= \left( 1-F_v(p_T) \right) \left(\frac{1}{f_v(p_T)} - \frac{T-1}{T } \frac{1}{\mathbf{1}^T E^{-1} \mathbf{1} } \right).
\end{align*}
In addition, the optimal normalized revenue is
\begin{align}\label{eq:NonUniformValuationsRevenue}
p_T \left(1- F_v(p_T) \right)+ \frac{T-1}{2 T} \frac{1}{\mathbf{1}^T E^{-1} \mathbf{1}} \left(1- F_v(p_T) \right)^2. 
\end{align}
\end{theorem}
\begin{proof}
We provide the main steps of the proof for $T=2$ and uniform valuations. The complete proof is presented in the Appendix. 

We first characterize the buyers' equilibrium given a price sequence as a function of network matrix $E$ and then find the optimal price sequence. Using Corollary \ref{Cor:SymmetricEquil}, the indifference condition for any block $h \in [m]$ in the first period becomes 
\begin{align*}
v^{(h)}_2- p_2 &= v^{(h)}_2-p_1+ \sum_{k \in [m]} \frac{E_{hk}}{n} \mathbb{E}\left[ \sum_{i \in G_k} \mathbf{1}\{i \text{ buys at } p_2 \} \right]  = v^{(h)}_2-p_1+ \sum_{ k \in [m]} E_{hk} \alpha_k (1- v^{(k)}_2),
\end{align*}
where we used uniform distribution for valuations in the last equality. 
By letting $\mathbf{v_2}=(v^{(1)}_2, \dots, v^{(m)}_2)$ and using the definition of $A$ and $E$, we can write this equation in compact form as 
\begin{align}\label{eq:prooftemp1}
 \mathbf{v_2} = \mathbf{1} - \left( EA \right)^{-1} \mathbf{1} (p_1- p_2).
\end{align}
In the second period, any remaining buyer $i$ in block $G_h$ buys the item if and only if 
\begin{align*}
v^{(i)}- p_1 + \sum_{k \in [m]} \sum_{j \in G_k} \frac{E_{hk}}{n}\mathbf{1}\{v^{(j)} \ge v_2^{(k)}\} \ge 0,
\end{align*} 
which for uniform valuations happens with probability 
\begin{align*}
\mathcal{P}_{[0,1]} \left(1- \frac{p_1 - \sum_{k \in [m]} \sum_{j \in G_k} \frac{E_{hk}}{n} \mathbf{1}\{v^{(j)} \ge v_2^{(k)}\}}{v_2^{(h)}} \right).
\end{align*}
Therefore, the monopolist's expected revenue can be written as
\begin{footnotesize}
\begin{align}\label{eq:revexptected}
\sum_{1 \le k_h \le n_h} & \left( \prod_{h=1}^m \binom{n_h}{k_h} (1- v_2^{(h)})^{k_h} (v_2^{(h)})^{n_h-k_h} \right)  \left( p_2 \sum_{h=1}^m k_h + p_1 \sum_{h=1}^m (n_h - k_h) \mathcal{P}_{[0,1]} \left(1- \frac{p_1 - \sum_{h' \in [m]}  E_{hh'} \frac{ k_{h'}}{n} }{v_2^{(h)}} \right) \right),
\end{align}
\end{footnotesize}
where the first term of each summand is the probability of a multivariate Binomial random variable capturing the probability of the event that in the first period for each block $h \in [m]$, $k_h$ out of $n_h$ many buyers purchase the item. The second term of each summand is the expected revenue of the monopolist given this event. In particular, for each block $h$, $k_h$ buyers purchase the product at price $p_2$ and each of the remaining $(n_h-k_h)$ buyers purchase the product at price $p_1$ with probability $\mathcal{P}_{[0,1]} \left(1- \frac{p_1 - \sum_{h' \in [m]}  E_{hh'} \frac{ k_{h'}}{n} }{v_2^{(h)}} \right)$. 

To obtain a closed-form expression for the optimal expected revenue, we consider the limiting normalized revenue as $n \to \infty$ and then use Bernstein polynomial convergence presented in Theorem \ref{Thm:BersteinSimplex}. Letting 
\begin{align*}
f\left(\frac{k_1}{n_1}, \dots, \frac{k_m}{n_m} \right) = \left( p_2 \sum_{h=1}^m \frac{k_h}{n_h} \alpha_h  + p_1 \sum_{h \in [m]} \left( 1 - \frac{k_h}{n_h} \right) \alpha_h \mathcal{P}_{[0,1]} \left(1- \frac{p_1 - \sum_{h'\in [m]} E_{hh'} \alpha_{h'} \frac{k_{h'}}{n_{h'}} }{v_2^{(h)}} \right) \right),
\end{align*}
the normalized expected revenue given in Eq. \eqref{eq:revexptected}  becomes 
\begin{align*}
\sum_{1 \le k_h \le n_h} \left( \prod_{h=1}^m \binom{n_h}{k_h} (1- v_2^{(h)})^{k_h} (v_2^{(h)})^{n_h-k_h} \right)   f\left(\frac{k_1}{n_1}, \dots, \frac{k_m}{n_m} \right).
\end{align*}
Therefore, using Theorem \ref{Thm:BersteinSimplex} the limiting normalized revenue becomes
\begin{align}\label{eq:prooftemp3}
& p_2 \sum_{h=1}^m \alpha_h (1- v_2^{(h)})+ \sum_{h} p_1 (\alpha_h - (1-v^{(h)}_2) \alpha_h) \mathcal{P}_{[0,1]} \left(1- \frac{p_1 - \sum_{h'} E_{hh'} (1-v^{(h')}_{2}) \alpha_{h'}}{v_2^{(h)}} \right)  \nonumber \\
& = p_2 +(p_1- p_2) (\alpha^T \mathbf{v}_2) - p_1^2 + p_1 \alpha^T \left( EA \right) (\mathbf{1} - \mathbf{v}_2),
\end{align}
where the second equality follows from Assumption \ref{assump:generalnetworkNonUniformValuations} and Eq. \eqref{eq:prooftemp1}, as we will show at the end of this proof. Combining Eqs. \eqref{eq:prooftemp1} and \eqref{eq:prooftemp3}, the normalized limiting revenue denoted by $h(p_1, p_2)$ becomes 
\begin{align*}
h(p_1, p_2) & = p_2 +(p_1- p_2) \left(\alpha^T \left( \mathbf{1} - \left( EA \right)^{-1} \mathbf{1} (p_1- p_2)\right) \right) - p_1^2 + p_1 \alpha^T \mathbf{1}(p_1- p_2) \\
& = p_1 - \left(\mathbf{1}^T E^{-1} \mathbf{1} \right) (p_1- p_2)^2 - p_1 p_2.
\end{align*}
We next find $(p_1, p_2)$ that maximizes $h(p_1, p_2)$. The first order conditions of $h(p_1, p_2)$ are
\begin{align*}
\frac{\partial h(\cdot)}{\partial p_1} =  1- 2 (p_1- p_2) \left(\mathbf{1}^T E^{-1} \mathbf{1} \right)- p_2=0, \quad \frac{\partial h(\cdot)}{\partial p_2} =  2 (p_1- p_2) \left(\mathbf{1}^T E^{-1} \mathbf{1} \right)- p_1=0,
\end{align*}
which leads to 
\begin{align}\label{eq:pricetwoseq}
p_1= \frac{2 \left(\mathbf{1}^T E^{-1} \mathbf{1} \right)}{4 \left(\mathbf{1}^T E^{-1} \mathbf{1} \right)-1}, \quad p_2= \frac{2 \left(\mathbf{1}^T E^{-1} \mathbf{1} \right)-1}{4 \left(\mathbf{1}^T E^{-1} \mathbf{1} \right)-1}.
\end{align}
The solution of the first order conditions indeed lead to the maximum of $h(p_1, p_2)$. We can see this by taking second order derivative of $h(p_1, p_2)$ and showing that the Hessian is negative semidefinite. In particular, taking second order derivative of $h(p_1, p_2)$ leads to the following Hessian
\begin{align*}
M = \begin{pmatrix} 
 -2  \left(\mathbf{1}^T E^{-1} \mathbf{1} \right) && 2  \left(\mathbf{1}^T E^{-1} \mathbf{1} \right)-1  \\\\
2  \left(\mathbf{1}^T E^{-1} \mathbf{1} \right)-1 && -2  \left(\mathbf{1}^T E^{-1} \mathbf{1} \right)
\end{pmatrix},
\end{align*}
which is negative definite. This can be seen by noting that 
$M_{11} < 0$ and $\text{det}(M)= 4 \left(\mathbf{1}^T E^{-1} \mathbf{1} \right) -1 > 0$, which holds because Assumption \ref{assump:generalnetworkNonUniformValuations} for uniform distributions guarantees $\mathbf{1}^T E^{-1} \mathbf{1}  \ge \frac{1}{2}$.\footnote{A symmetric $2 \times 2$ matrix $M$ is negative definite if $M_{11} < 0$ and $\text{det}(M) > 0$.} Therefore, the maximum of $h(p_1, p_2)$ is attained at the solution of the first order conditions. Plugging $p_1$ and $p_2$ from Eq. \eqref{eq:pricetwoseq} in $h(p_1, p_2)$, the optimal revenue becomes
\begin{align*}
\frac{\left(\mathbf{1}^T E^{-1} \mathbf{1} \right)}{4 \left(\mathbf{1}^T E^{-1} \mathbf{1} \right)-1}. 
\end{align*}
Finally, we show that the prices $p_1$ and $p_2$ and their corresponding $\mathbf{v}_2$ guarantee $1- \frac{p_1 - \sum_{h'} E_{hh'} (1-v^{(h')}_{2}) \alpha_{h'}}{v_2^{(h)}} \in [0,1]$ for all $h \in [m]$, showing that the projection operators used in Eq. \eqref{eq:prooftemp3} are identity. This is equivalent to 
\begin{align}\label{eq:proofboundary}
0 \le p_1 - \sum_{h'} E_{hh'} (1-v^{(h')}_{2}) \alpha_{h'} \le v_2^{(h)}, \quad \text{ for all }h \in [m].
\end{align}
Using Eq. \eqref{eq:prooftemp1}, we can rewrite Eq. \eqref{eq:proofboundary} in vector form as $\mathbf{0} \le p_1 \mathbf{1}- \left( EA \right)\left(\mathbf{1} - \mathbf{v}_2 \right) \mathbf{1} \le \mathbf{v_2}$. 
The lower bound evidently holds as $p_1 \mathbf{1}- \left( EA \right)\left(\mathbf{1} - \mathbf{v}_2 \right)= p_2 \mathbf{1} \ge 0$. The upper bound after plugging in $p_2$ and $p_1$ in Eq. \eqref{eq:prooftemp1} and finding $\mathbf{v}_2$, becomes
\begin{align*}
\frac{2 \left(\mathbf{1}^T E^{-1} \mathbf{1} \right)-1 }{4 \left(\mathbf{1}^T E^{-1} \mathbf{1} \right)-1} \mathbf{1} \le \mathbf{1} - \left(E A \right)^{-1} \mathbf{1} \frac{1}{4 \left(\mathbf{1}^T E^{-1} \mathbf{1} \right)-1},
\end{align*}
which holds because of Assumption \ref{assump:generalnetworkNonUniformValuations} and in particular $\left(E A \right)^{-1} \mathbf{1} \le \left(\mathbf{1}^T E^{-1} \mathbf{1} \right) \mathbf{1}$. This completes the proof for the special case of $T=2$ and uniform valuations. 
\end{proof}
Theorem \ref{thm:NonUniformValuations} shows that the optimal price sequence is linearly increasing. Moreover, the impact of the network matrix on the optimal price sequence and the optimal revenue is captured by a single network measure $\frac{1}{\mathbf{1}^T E^{-1} \mathbf{1}}$.
\begin{definition}\label{def:aggE}
\textup{
 We refer to the term $\frac{1}{\mathbf{1}^T E^{-1} \mathbf{1}}$ as \emph{\NE} and denote it by $\Psi(E)$.
}
\end{definition}
We next show the properties of the optimal normalized revenue as well as the optimal price sequence as a function of the \NE~and the number of rounds.
\begin{proposition}\label{Lem:PsiIncreasing}
Suppose Assumption \ref{assump:generalnetworkNonUniformValuations} holds. The optimal normalized revenue is increasing in $T$ and increasing and convex in \NE~$\Psi(E)$. Moreover, the slope of the optimal price sequence (i.e., the difference of the optimal prices in two consecutive rounds) is increasing in $\Psi(E)$.
\end{proposition}
Theorem \ref{thm:NonUniformValuations} and Proposition \ref{Lem:PsiIncreasing} yield the following insights:
\begin{figure}[t]
    \centering
    \begin{subfigure}[b]{0.4\textwidth}
        \includegraphics[width=\textwidth]{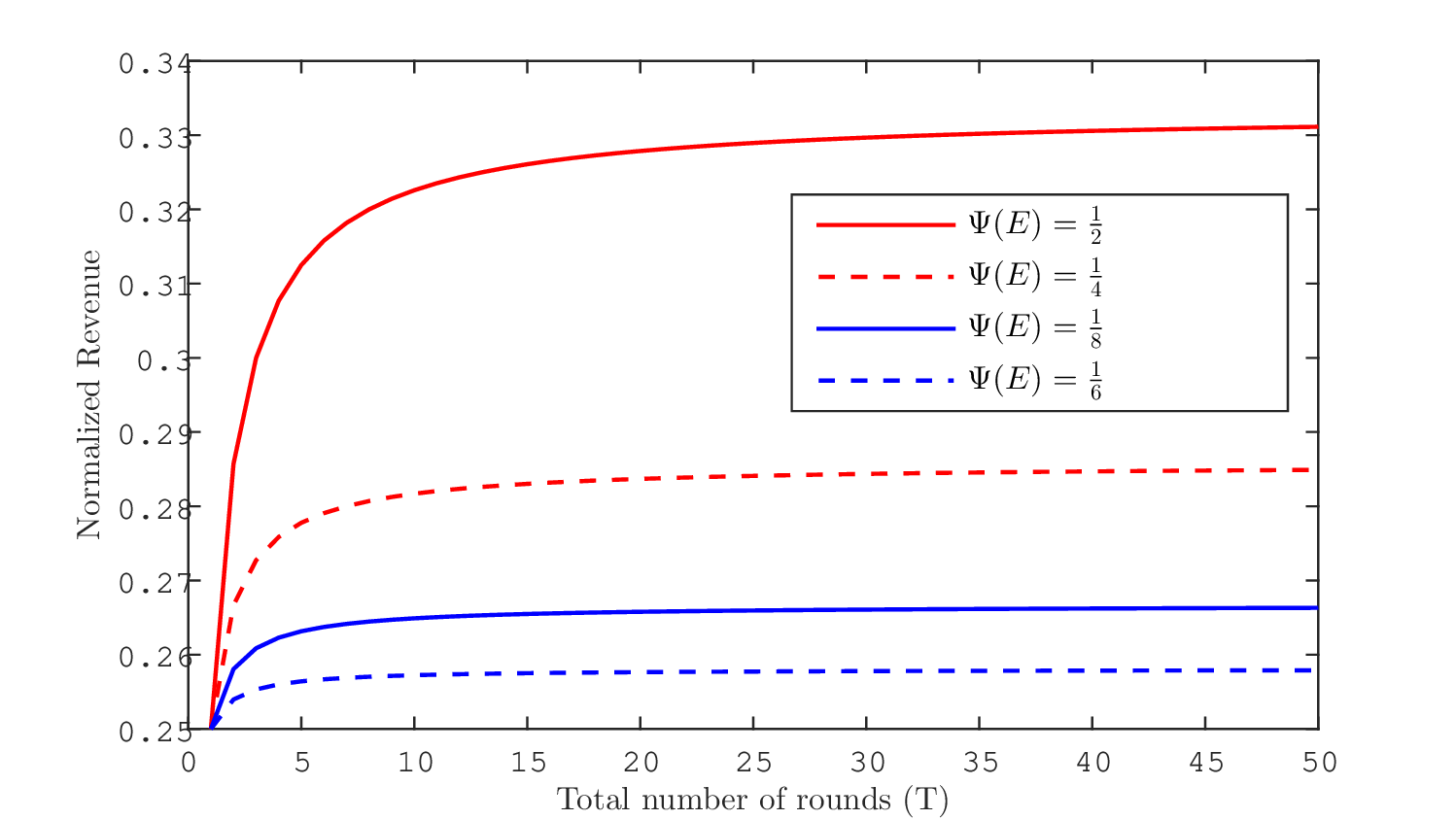}
        \caption{}
         \label{fig:GeneralGraphTperiodWithCommitmentRevenue}  
    \end{subfigure}
    ~
    \begin{subfigure}[b]{0.4\textwidth}
        \includegraphics[width=\textwidth]{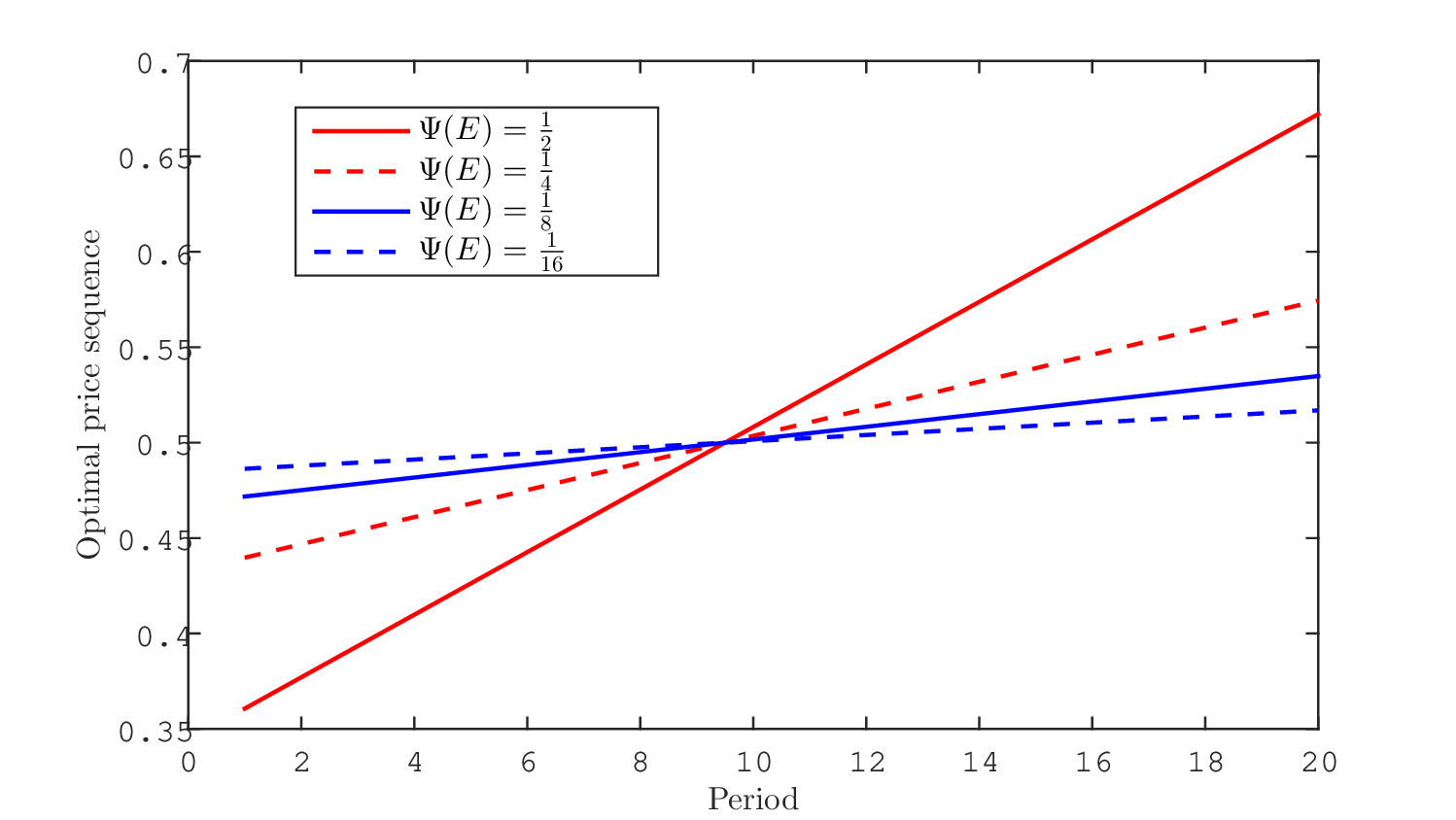}
       \caption{}
 \label{fig:GeneralGraphTperiodWithCommitmentPricepath}
    \end{subfigure}
    \caption{(a) optimal normalized revenue as a function of the number of rounds for uniform distribution and various \NE s $\Psi(E)$ and (b) optimal price sequence for uniform distribution, $T=20$, and various \NE s $\Psi(E)$. }
\end{figure}
\begin{itemize}

\item Proposition \ref{Lem:PsiIncreasing} shows that the optimal normalized revenue is an increasing convex function in \NE. Intuitively, this holds because increasing $\Psi(E)$ increases the utility of buyers in two different ways: (i) \emph{direct effect}: as $\Psi(E)$ increases, keeping purchase probability of other buyers the same, each buyer enjoys a higher network effect, and  (ii) \emph{indirect effect}: as $\Psi(E)$ increases, the purchase probability of other buyers in previous rounds increases (i.e., the critical thresholds decrease), leading to higher utility. 
\item Figure \ref{fig:GeneralGraphTperiodWithCommitmentRevenue} shows the optimal normalized revenue as a function of the number of rounds, illustrating that the optimal revenue is increasing in $T$. The optimal revenue approaches the limiting normalized revenue as $T \to \infty$ relatively fast. For instance, for uniform valuations if we want to obtain $q = 95\%$ of the optimal revenue (i.e., revenue of infinitely many rounds), then for $\Psi(E)= \frac{1}{5}$, $T= 3$ rounds suffice and for $\Psi(E)=\frac{4}{5}$, $T=13$ rounds suffice.  
\item Figure \ref{fig:GeneralGraphTperiodWithCommitmentPricepath} shows the optimal price sequence for various \NE s, illustrating that the slope of the optimal price sequence is increasing in $\Psi(E)$. Proposition \ref{Lem:PsiIncreasing} shows the slope of the optimal price sequence is increasing in the \NE. Intuitively, this holds because as $\Psi(E)$ increases, the past purchases contribute more to the utility of buyers, incentivizing them to purchase at a higher price. 
\item Buyers with valuations below the critical threshold of the last period $\mathbf{v}_1$ do not buy the item. This threshold is decreasing in $\Psi(E)$ and $T$. Therefore, increasing $\Psi(E)$ or $T$ increases the number of users who purchase the item.
\end{itemize}


\section{Aggregate Network Effect}\label{sec:networkeffects}
In Theorem \ref{thm:NonUniformValuations} we characterized the optimal price sequence as a function of the \NE~$\Psi(E)$. We now study in more detail the effect of network properties on $\Psi(E)$ and the optimal price sequence and revenue. In this regard, we consider a weakly-tied block model, where $E=I+\delta C$ for $C \ge 0$ and some sufficiently small $\delta \ge 0$. This network matrix represents a natural setting in which the utility gain of buyers within each block is larger than the ones across blocks. Also, note that for this network matrix, using Lemma \ref{Lem:DeltaSmall}, for uniform valuations and $A= \frac{1}{m}I$, Assumption \ref{assump:generalnetworkNonUniformValuations} and therefore Theorem \ref{thm:NonUniformValuations} holds. In particular, a second order Taylor approximation of $\Psi(E)$ leads to
\begin{align}\label{eq:remark:approx1:second}
\frac{1}{m}+ \delta \frac{\mathbf{1}^T C \mathbf{1}}{m^2} + \delta^2 \frac{\left( \mathbf{1}^T C \mathbf{1} \right)^2 - m \mathbf{1}^T C^2 \mathbf{1}}{m^3}. 
\end{align}
Proposition \ref{Lem:PsiIncreasing} together with Eq. \eqref{eq:remark:approx1:second} leads to the following implications.
\begin{itemize}
\item The second term shows that higher sum of inter blocks utility gains, i.e, $\mathbf{1}^T C \mathbf{1}$, leads to higher $\Psi(E)$ and therefore higher revenue. 
\item The third term shows that for a given $\mathbf{1}^T C \mathbf{1}$, the highest revenue is obtained for a network with the minimum $\mathbf{1}^T C^2 \mathbf{1}$ where
\begin{align*}
\mathbf{1}^T C^2 \mathbf{1} =\sum_{i, j=1}^m \sum_{k=1}^m C_{ik} C_{kj} = \sum_{i, k=1}^m C_{ik} d_k^{(\text{out})} = \sum_{k=1}^m d_k^{(\text{out})}  d_k^{(\text{in})}.
\end{align*}
Here $d_k^{(\text{out})}$ and $d_k^{(\text{in})}$ correspond to the out-degree and in-degree in the network matrix $C$, respectively. 
This shows that the highest revenue is obtained for a network with minimum $\sum_{k} d_k^{(\text{out})}  d_k^{(\text{in})}$. With a small $\sum_{k} d_k^{(\text{out})}  d_k^{(\text{in})}$, the influential blocks (i.e., with high out-degrees) are less influenced by other blocks (i.e., have low in-degrees). Therefore, the influential blocks have little incentive to postpone their purchase and prefer to purchase earlier (i.e., they have a low critical threshold). This in turn incentivizes users from other blocks to purchase at higher prices in subsequent periods (because of the higher network effect term), increasing the revenues of the monopolist. As an example, for $C \in \{0,1\}^{m \times m}$ (i.e., $C_{ij}=1$ if there exists an edge from $j$ to $i$), a bipartite directed graph has the highest revenue ($d_i^{\text{(in)}} d_i^{\text{(out)}}=0$ for all $i \in [m]$). 

We next show that the minimum $\sum_{k=1}^m d_k^{(\text{out})}  d_k^{(\text{in})}$ (for a given in-degree sequence or fixed $\mathbf{1}^T C \mathbf{1}$) is obtained for a network with the most degree sequence ``imbalance'', where the imbalance sequence is defined as $(d_1^{\text{(in)}}- d_1^{\text{(out)}}, \dots, d_m^{\text{(in)}}- d_m^{\text{(out)}})$ (see \cite{mubayi2001realizing}). 
\begin{proposition}\label{Lem:Majorization}
Consider two network matrices $C$ and $\tilde{C}$ with the same in-degree sequence $d_1^{\text{(in)}} \ge \dots \ge d_m^{\text{(in)}}$ and different out-degree sequences denoted by $(c_1^{\text{(out)}}, \dots, c_m^{\text{(out)}})$ and $(\tilde{c}_1^{\text{(out)}}, \dots, \tilde{c}_m^{\text{(out)}})$, respectively. If the imbalance sequence of $C$ majorizes the imbalance sequence of $\tilde{C}$, then $\mathbf{1}^T C^2 \mathbf{1} \le \mathbf{1}^T \tilde{C}^2 \mathbf{1}$.\footnote{We say a sequence $(a_1, \dots, a_m)$ majorizes $(b_1, \dots, b_m)$ if and only if for any $1 \le i \le m$ we have $\sum_{j=1}^i a_i \ge \sum_{j=1}^i b_i$. } 
\end{proposition}

The next example illustrates the effect of imbalance sequence on revenue by comparing the revenues of a directed chain, a directed ring, and a directed star network. 
\begin{example}\label{Example:three}
\textup{
We consider a directed chain, a directed ring, and a directed star network (where the edges are from the periphery nodes to the center node) with $m=10$ blocks. For all of these networks we let the network matrix be $E=I+ \delta C$, $\delta=.29$, $\sum_{i, j}C_{ij}=30$, and the weight of different edges in each network be the same. The first and second terms of Eq. \eqref{eq:remark:approx1:second} are the same for these three networks, but the third terms are different. In particular, the term $\sum_{k=1}^m d_k^{(\text{out})}  d_k^{(\text{in})}$ for directed ring, directed chain, and directed star are $81.8$, $81$, and $0$, respectively. Therefore, the \NE~and the revenue of directed star is higher than directed chain, which is higher than directed ring. Figure \ref{fig:NE_fig1} illustrates the revenues of these three networks as a function of the number of periods. 
}
\end{example}
\begin{figure}[t]
    \centering
        \includegraphics[width=.5\textwidth]{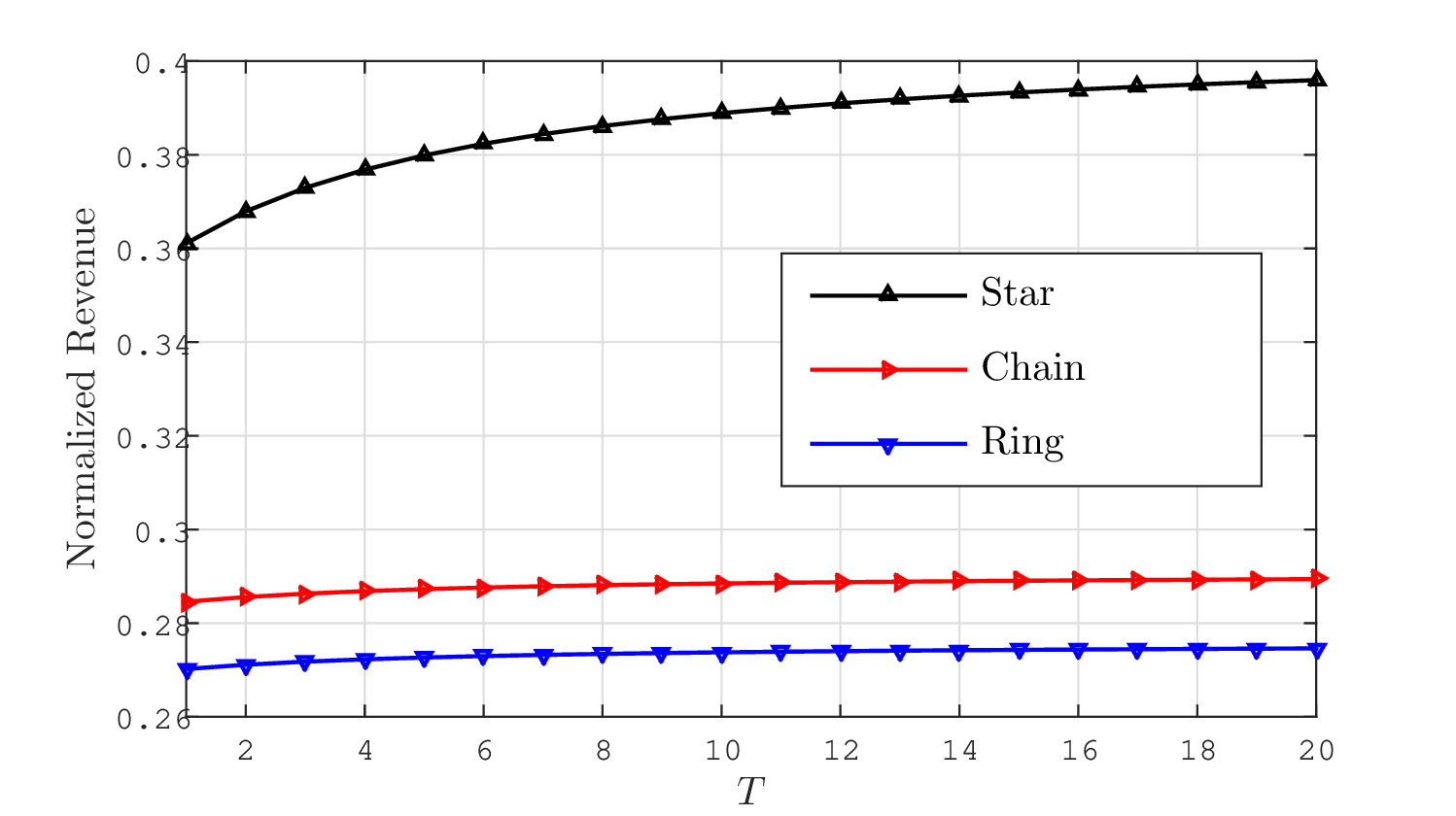}
    \caption{The normalized revenue as function of the number of periods for the networks of Example \ref{Example:three}:  directed chain, directed ring, and directed star networks. 
 }\label{fig:NE_fig1}
\end{figure}
\item For a given $\mathbf{1}^T C \mathbf{1}$, among balanced networks for which $d_k^{\text{(in)}} = d_k^{\text{(out)}}$, for all $k \in [m]$ (note that a symmetric network, i.e., $C_{ij}=C_{ji}$ is balanced) the highest revenue is obtained by a regular network, i.e., a network with $d_j^{\text{(out)}} = d_k^{\text{(out)}}$ for all $j,k \in [m]$. This follows because we have 
\begin{align*}
\mathbf{1}^T C^2 \mathbf{1} = \sum_{i, k=1}^m C_{ik} \left( \sum_{j=1}^m C_{kj} \right) = \sum_{k=1}^m \left( d_{k}^{\text{out}} \right)^2 \ge \frac{1}{m} \left( \sum_{k=1}^m d_{k}^{\text{out}} \right)^2,
\end{align*}
where we used the balancedness in the last equality and Cauchy-Schwarz in the inequality. The inequality becomes an equality if and only if for any $k$, we have $d_k^{\text{(out)}} = \frac{1}{m}\sum_{i,j=1}^m C_{ij}$, showing that a regular network has the highest revenue. Here, because of the balancedness, each block is equally influential and influenced by others. Now suppose the network is not regular and consider the block with the highest out-degree (which is equal to its in-degree). The customers in this block have a high in-degree and therefore prefer to postpone their purchase to future rounds (because they will obtain a large utility gain from other purchases). Thus, the monopolist cannot fully utilize the network effect of this block on other blocks to increase the revenue. 
\end{itemize}
\section{Price Discrimination}\label{sec:PriceDiscrimination}
In this section, we study price discrimination and characterize the optimal price sequence. With price discrimination, at each round the price offered to different blocks can be different. Throughout this section we assume that the valuations are uniform, but, the results can be generalized to any valuation distribution under proper regularity assumptions. 
\begin{proposition}\label{thm:pricediscrimination}
Suppose the valuations are uniform and $I-EA$ is positive semidefinite.
 The optimal price sequence in the limit (as $n \to \infty$) is given by
\begin{align}\label{eq:thm:pricediscrimination:pricesequence}
\mathbf{p}_t= \frac{T-t}{T} \left(I- \frac{T-1}{T} EA \right)^{-1} EA \mathbf{p}_T + \mathbf{p}_T , \quad t=T, \dots, 1,
\end{align}
where $\mathbf{p}_T$ is the solution of 
\begin{align}\label{eq:thm:pricediscrimination:pricesequence2}
\mathbf{p}_T= \frac{1}{2} \mathbf{1} - \frac{T-1}{4T} \left(I - \frac{T-1}{2T} EA \right)^{-1} EA \mathbf{1}.
\end{align}
\end{proposition}
Note that for a symmetric weakly-tied block model, i.e., $E= I + \delta C$, similar to Lemma \ref{Lem:DeltaSmall}, for small enough $\delta$, Assumption \ref{assump:generalnetworkNonUniformValuations} holds. 

%
We can rewrite the optimal price sequence as 
\begin{align*}
\mathbf{p}_t=  \frac{T-t}{T-1} \mathbf{b}\left(  EA, \frac{T-1}{2T} \right) - \frac{1}{T-1} \mathbf{1} + \mathbf{p}_T, \quad t=T, \dots, 1,  
\end{align*}
where $\mathbf{b}(E, \beta)$ is the Bonacich centrality with parameter $\beta$ in network $E$ defined as $\mathbf{b}(E, \beta)= \left(I - \beta E \right)^{-1} \mathbf{1}$ (see \cite{bonacich1987power}). 
Proposition \ref{thm:pricediscrimination} establishes that the optimal price sequence is linear and the slope of this linear price sequence is given by a ``weighted Bonacich centrality''. This implies that in the early periods the monopolist offers a lower price to more central blocks and rapidly increases the price offered to them in subsequent periods. This is to encourage more central buyers to purchase in the early periods which in turn incentivizes more buyers (due to larger centrality) to purchase in subsequent periods. 

We next compare the optimal price sequence in this setting with that of the static pricing studied in \cite{candogan2012optimal}. The buyer's strategies in a setting with one round of pricing is given by a threshold rule where buyers in block $i$ purchase the item if their valuation exceed $v^{(i)}$. These thresholds satisfy the indifference condition
\begin{align*}
\mathbf{v}- \mathbf{p}+ EA \left(\mathbf{1} - \mathbf{v} \right)=\mathbf{0},
\end{align*}
where $\mathbf{v}=\left(v^{(1)}, \dots, v^{(m)} \right)$. This leads to $\mathbf{v}= \left( I - EA \right)^{-1} \left(\mathbf{p} - EA \mathbf{1} \right)$.\footnote{For any $k \in [m]$ the expected normalized number of customers in block $k$ who purchase the item is $1- v^{(k)}$. Therefore, our model becomes identical to the consumption model of \cite{candogan2012optimal} by letting the ``consumption level'' of block $k$ be equal to $1-v^{(k)}$.} 
Therefore, the normalized revenue becomes
\begin{align}\label{eq:OptRevSingleRound}
\mathbf{p}^T A \left(\mathbf{1} - \mathbf{v} \right) 
= \mathbf{p}^T A \left(I - EA \right)^{-1} \mathbf{1} - \mathbf{p}^T A \left(I - EA \right)^{-1} \mathbf{p}.
\end{align}
Maximizing Eq. \eqref{eq:OptRevSingleRound} over $\mathbf{p}$ leads to 
\begin{align}\label{eq:OptPriceSingleRound}
\mathbf{p} = \mathbf{1}- \frac{1}{2} \left(I - EA \right) \left(I - \frac{EA +(EA)^T}{2} \right)^{-1} \mathbf{1},
\end{align}
which is the same as the optimal prices obtained for static pricing in \citet[Theorem 2]{candogan2012optimal} with $G=EA$. 
Comparing our result in Proposition \ref{thm:pricediscrimination} and Eq. \eqref{eq:OptPriceSingleRound}, we note that the optimal price given in Eq. \eqref{eq:OptPriceSingleRound} for symmetric networks (i.e., $EA=(EA)^T$) is $\mathbf{p}= \frac{1}{2} \mathbf{1}$ (see also \citet[Corollary 1]{candogan2012optimal}).
However, in our setting with more than one period, even for symmetric networks, the optimal price sequence depends on the network structure (see Eqs. \eqref{eq:thm:pricediscrimination:pricesequence} and \eqref{eq:thm:pricediscrimination:pricesequence2}) and the monopolist obtains revenue gains from the network effects which is in sharp contrast with \citet[Section 5]{candogan2012optimal}.

\section*{Acknowledgment}
We thank participants at several seminars and conferences for useful suggestions and comments. We specially thank Daron Acemoglu, Kimon Drakopolous, Saeed Alaei, and Ozan Candogan for very helpful discussions and suggestions.
\section{Conclusion}\label{sec:Conclusion}
We study the problem of finding the optimal price sequence for a product given a set of customers with heterogeneous
valuations that strategically decide their purchase time (if any). The product features network effects, i.e., the utility of each customer depends on the price sequence, her valuation, as well as a weighted number of other buyers who have purchased the item. We establish that the problem of finding the optimal price sequence is a tractable one and explicitly characterize the optimal price sequence as a function of the network structure. Our main result identifies a novel dependence on the network structure: sum of the entries of the inverse of network matrix, termed \NE. From a structural
perspective, the optimal price sequence is always linearly increasing with a slope that is increasing in the \NE. We establish that increasing network imbalance increases the \NE~which, in turn, increases the revenue.
The framework and results we present in this paper lay the ground for a potential new approach to the class of dynamic pricing problems with strategic customers and combinatorial structures. Avenues for future research include the expansion of the set of problems that may be tackled through the present approach. For example, the question of dynamic strategic pricing with limited inventory and strategic buyers and seller is a natural extension.
\section*{Appendix}\label{sec:app}
\subsection*{Proof of Proposition \ref{Pro:SingleCrrossing}}\label{App:Proof:Pro:SingleCrrossing}
\textbf{Proof of part (a):} if $p_{s+1} > p_{s} $ for some $s=T-1, \dots, 1$, then none of the buyers will buy the item at round $s+1$ (note that the period with price $p_s$ is after the period with price $p_{s+1}$). This is because if they wait until the next round, the price decreases and the network effect term of their utility weakly increases (it either remains the same or increases). Therefore, the equilibrium path (the purchase decision of buyers) in the continuation game (i.e., in rounds $s+1, s, \dots, 1$) is the same as a game in which we have $p_s=p_{s+1}$. Staring from the last period (i.e., $s=1$) and repeatedly applying this argument shows the existence of a non-decreasing price sequence with the same equilibrium path as the optimal price sequence. 
\\\textbf{Proof of part (b):} for a given equilibrium, suppose that buyer $i$ with valuation $v$ finds it optimal to purchase at price $p_t$ in round $t$. Then it must be the case that her utility from purchasing in round $t$ is not smaller than her expected utility from postponing the purchase to future rounds (i.e., not purchasing at round $t$). Therefore, we must have 
\begin{align*}
& v- p_t + \sum_{j} g_{ij} \mathbf{1} \{j \text{ bought at } s > t \} \\
& \ge \mathbb{E}\left[ \sum_{s=t-1}^1 \mathbf{1}\{i \text{ buy at price } p_s\} \left( v- p_s + \sum_{j} g_{ij} \mathbf{1} \{j \text{ bought at } s' > s \} \right) \mid H_t \right] \\
& = \sum_{s=t-1}^1 \mathbb{P}\left[ i \text{ buys at } s \right] \left( v- p_s + \mathbb{E}\left[ \sum_{j} g_{ij} \mathbf{1} \{j \text{ bought at } s' > s \} \mid H_t \right]\right).
\end{align*}
Since $\sum_{s=t-1}^0 \mathbb{P}\left[ i \text{ buys at } s \right] \le 1$ (there is a probability with which $i$ does not buy at all), the derivative in $v$ of the left hand side of this inequality is at least as large as that of the right hand side. Therefore, buyer $i$ with valuations $v' > v$ finds it strictly optimal to also purchase at round $t$ with price $p_t$. This shows that if buyer $i$ at time $t$ does not purchase, then her valuation is larger than a certain threshold denoted by $v_t^{(i)}$. Finally, note that each buyer that purchases the item leaves the game, leading to $v_t^{(i)} \ge v_{t-1}^{(i)}$, for $2 \le t\le T $. 
\subsection*{Proof of Lemma \ref{Lem:buyer'sbehavior}}\label{App:Proof:Lem:buyer'sbehavior}
Using Proposition \ref{Pro:SingleCrrossing}, for a given price sequence $p_T, \dots, p_1$ and history $H^t$, the decision of buyer $i \in [n]$ at time $t$ is to buy if and only if her valuation exceeds $v^{(i)}_t$. 
 Buyer $i$ with valuation $v^{(i)}_t$ must be indifferent between accepting price $p_t$ and waiting until the subsequent round (in a continuation game with $t-1$ periods to go). If buyer $i$ with valuation $v^{(i)}_t$ purchases at price $p_t$, her utility (conditional on period $t$ having been reached) is
\begin{align}\label{eq:periodtpayoff2}
v^{(i)}_t - p_t + \sum_{j} g_{ij} \mathbf{1} \{j \text{ has bought at } s > t \} .
\end{align}
By waiting one more period instead, buyer $i$ with valuation $v^{(i)}_t$ obtains utility
\begin{align}\label{eq:periodtpayoff3}
v^{(i)}_t - p_{t-1} + \sum_{j} g_{ij} \mathbb{E} \left[ \mathbf{1} \{j \text{ has bought at } s > t-1 \} \right].
\end{align}
Buyer $i$ with threshold $v^{(i)}_t$ at round $t$ must be indifferent between buying at round $t$ and buying at the next round, i.e., round $t-1$. Subtracting Eq. \eqref{eq:periodtpayoff2} and Eq. \eqref{eq:periodtpayoff3} leads to 
\begin{align*}
v^{(i)}_t- p_t = v^{(i)}_{t}- p_{t-1} +  \sum_{j \in [n] \setminus H^t \setminus \{i\}} g_{ij} \mathbb{P}\left[ j \text{ buys at  round } t \right],
\end{align*}
where  
\begin{align*}
\mathbb{P}\left[ j \text{ buys at  round } t \right]= \mathbb{P}\left[v^{(j)} \ge v^{(j)}_{t} \mid v^{(j)} \le v^{(j)}_{t+1} \right], 
\end{align*}
Finally, note that if $v_{t+1}^{(j)}= 0$, then buyer $j$ has purchased the item in round $t+1$ and if $v_{t+1}^{(j)} \in (0,1]$, then using Bayes' rule we have 
\begin{align*}
\mathbb{P}\left[v^{(j)} \ge v^{(j)}_{t} \mid v^{(j)} \le v^{(j)}_{t+1} \right] = 1- \frac{F_v\left( v^{(j)}_t \right)}{F_v(v^{(j)}_{t+1})},
\end{align*}
which completes the proof.  
\subsection*{Proof of Lemma \ref{Lem:DeltaSmall}}\label{App:Proof:Lem:DeltaSmall}
\textbf{Proof of $(EA)^{-1} \mathbf{1} \ge \mathbf{0}$: } we will show that $E^{-1} \mathbf{1} \ge \mathbf{0}$. First note that since $\rho(C) \le ||C||_{\infty} \le \mathbf{1}^T C \mathbf{1}$, we have $\delta \rho(C) < 1$ which guarantees $I + \delta C$ is invertible. Since $\delta \rho(C) < 1$, we also have 
\begin{align*}
E^{-1}\mathbf{1} & =\left(I -\delta C + \delta^2 C^2 - \dots \right) \mathbf{1}= \left(I+ \delta^2 C^2 + \dots \right)\left(\mathbf{1}- \delta C \mathbf{1} \right) \ge \mathbf{1}- \delta C \mathbf{1} \ge \mathbf{0},
\end{align*}
where the last two inequalities follow from $\delta ||C||_{\infty} \le 1$. 
\\ \textbf{Proof of $\frac{\left(E A \right)^{-1} \mathbf{1}}{\mathbf{1}^T E^{-1} \mathbf{1}  } \le  \mathbf{1}$:}
using Taylor series expansion of $(I+\delta C)^{-1}$ (which converges because $\delta \rho(C)<1$), leads to 
\begin{align}\label{eq:ineqtemp1}
& \mathbf{1}^T \left(I + \delta C \right)^{-1} \mathbf{1} \ge \mathbf{1}^T  \left(I - \delta C \right) \mathbf{1} = m - \delta \left( \mathbf{1}^T C \mathbf{1} \right).
\end{align}
We also have 
\begin{align}\label{eq:ineqtemp2}
E^{-1} \mathbf{1} & = \left(I - \delta C + \delta^2 C - \delta^3 C^3 + \dots \right) \mathbf{1}  = \left(I - \delta C + \delta^2 C  \right) \mathbf{1} - \left( \delta^2 C^2 + \delta^4 C^4 + \dots \right) \left(I -\delta C \right)\mathbf{1} \nonumber \\
& \stackrel{(1)}{\le} \left(I - \delta C + \delta^2 C  \right) \mathbf{1}  \stackrel{(2)}{\le} \mathbf{1},
\end{align}
where inequality (1) follows from $\delta ||C||_{\infty} < 1$ and inequality (2) follows from $\delta C^2 \mathbf{1} \le C \mathbf{1}$. This inequality holds because for any $i \in [m]$ we have 
\begin{align*}
\delta \sum_{j=1}^m [C^2]_{ij} & = \delta \sum_{j=1}^m \sum_{k=1}^m C_{ij} C_{kj} = \delta \sum_{k=1}^m C_{ik} \sum_{j=1}^m C_{kj} \le \delta ||C||_{\infty} \sum_{k=1}^m C_{ik} \le  \sum_{k=1}^m C_{ik}.
\end{align*}
Putting Eqs. \eqref{eq:ineqtemp1} and \eqref{eq:ineqtemp2} together leads to 
\begin{align*}
\frac{\left(E A \right)^{-1} \mathbf{1}}{\mathbf{1}^T E^{-1} \mathbf{1} } & \le \frac{m E^{-1} \mathbf{1}}{m- \delta \mathbf{1}^T C \mathbf{1}} \le \frac{m \mathbf{1}}{m- \delta \mathbf{1}^T C \mathbf{1}} \le  \mathbf{1},
\end{align*}
where we used $\delta \mathbf{1}^T C \mathbf{1} \le m$ in the last inequality that follows from $\delta ||C||_{\infty} \le 1$. 
\\ \textbf{Proof of $x-\frac{1-F_v(x)}{f_v(x)}-\frac{F_v(x)}{\mathbf{1}^T E^{-1} \mathbf{1}}$ non-decreasing:} for uniform distribution this condition is equivalent to having $\mathbf{1}^T E^{-1} \mathbf{1} \ge \frac{1}{2}$. Finally, note that this holds because from \eqref{eq:ineqtemp1} we obtain $\mathbf{1}^T E^{-1} \mathbf{1} \ge m-\frac{1}{2} \ge \frac{1}{2}$. 

\subsection*{Proof of Theorem \ref{thm:NonUniformValuations}}\label{App:Proof:thm:NonUniformValuations}
Throughput this proof we use the following notation for probability mass function of a multinomial random variable:
\begin{align*}
\Phi(k_1, \dots, k_m; n; p_1, \dots, p_m)= \binom{n}{k_1, \dots, k_m} (1- p_1- \dots- p_m)^{n-k_1-\dots- k_m} \prod_{i=1}^{m-1} p_i^{k_i},
\end{align*}
for all $(k_1, \dots, k_m)$ and $(p_1, \dots, p_m)$ such that $\sum_{i=1}^m k_i \le n$ and $\sum_{i=1}^m p_i \le 1$. 

For a given price sequence $\mathbf{p}=(p_T, \dots, p_1)$ and a general distribution, using Corollary \ref{Cor:SymmetricEquil}, the critical thresholds defining Buyers' equilibrium satisfy
\begin{align}\label{eq:tempNonUniform1}
EA \left(F\left( \mathbf{v}_{t+1} \right) - F\left( \mathbf{v}_t \right) \right) & = (p_{t-1}- p_{t}) \mathbf{1}, \quad t=T, \dots, 2, \nonumber \\
\mathbf{v}_1 - p_1 \mathbf{1} + EA \left( \mathbf{1}- F\left( \mathbf{v}_2 \right) \right) & = 0,
\end{align}
with the convention that $\mathbf{v}_{T+1}=1$ where $F_v(\mathbf{v})= (F_v(v^{(1)}), \dots, F_v(v^{(m)}))$.  
We will next find the limiting normalized revenue for a given price sequence. 
The normalized revenue is 
\begin{align}\label{eq:BersteinSimplex22}
& \frac{1}{n} \sum_{t=T}^{1} p_t \mathbb{E}\left[ \sum_{i=1}^m \sum_{j \in G_i} \mathbf{1}\{ v^{(j)} \in  [v^{(i)}_{t}, v^{(i)}_{t+1}) \} \right] \nonumber \\
&=  \sum_{\substack{ \sum_{s=1}^T k^{(s)}_j \le n_j \\ j \in [m], s \in [T]}}   \prod_{j=1}^m \Phi\left(k^{(1)}_j, \dots, k^{(T)}_j; n_j; F_v\left(v_{1}^{(j)} \right)-F_v\left(v_{0}^{(j)} \right), \dots, F_v\left(v_{T}^{(j)}\right)-F_v\left(v_{T-1}^{(j)} \right) \right) \nonumber \\
& \frac{1}{n} \left( p_T \sum_{j=1}^m k^{(T)}_j  + p_{T-1} \sum_{j=1}^m k^{(T-1)}_j  + \dots + p_1 \sum_{j=1}^m k^{(1)}_j \right),
\end{align}
where by convention $v_0^{(j)}=0$ and the multinomial distribution captures the number of  possibilities for selecting a partition of $[n_j]$ into $T+1$ subsets of size $k^{(1)}_j, \dots, k^{(T)}_j$, and $n_j - k^{(1)}_j - \dots - k^{(T)}_j$. Note that for $s=1, \dots, T$, $k^{(s)}_j$ shows the number of buyers in block $j$ who buy at price $p_s$ and the remaining number of buyers in block $j$ (i.e., $n_j - k^{(1)}_j - \dots - k^{(T)}_j$ many) decide not to buy the product.
We rewrite Eq. \eqref{eq:BersteinSimplex22} and take the limit as $n \to \infty$, resulting in 
\begin{align*}
\lim_{n \to \infty} & \sum_{\substack{ \sum_{s=1}^T k^{(s)}_j \le n_j \\ j \in \{2, \dots, m-1\}, s \in [T]}}  \Phi\left(k^{(1)}_j, \dots, k^{(T)}_j; n_j; F_v\left(v_{1}^{(j)} \right)-F_v\left(v_{0}^{(j)} \right), \dots, F_v\left(v_{T}^{(j)}\right)-F_v\left(v_{T-1}^{(j)} \right) \right)   \\
& \times  \frac{1}{n} \left( p_T \sum_{j=2}^m k^{(T)}_j  + p_{T-1} \sum_{j=2}^m k^{(T-1)}_j  + \dots + p_1 \sum_{j=2}^m k^{(1)}_j \right)\\
& \times \sum_{\substack{ \sum_{s=1}^T k^{(s)}_1 \le n_1 \\  s \in [T]}}   \Phi\left(k^{(1)}_1, \dots, k^{(T)}_1; n_1; F_v\left(v_{1}^{(1)} \right)-F_v\left(v_{0}^{(1)} \right), \dots, F_v\left(v_{T}^{(1)}\right)-F_v\left(v_{T-1}^{(1)} \right) \right)  \\
& \times  \frac{\alpha_1}{n_1} \left( p_T  k^{(T)}_1  + p_{T-1}  k^{(T-1)}_1  + \dots + p_1  k^{(1)}_1 \right) \\
&= \lim_{n \to \infty} \sum_{\substack{ \sum_{s=1}^T k^{(s)}_j \le n_j \\ j \in \{2, \dots, m-1\}, s \in [T]}}  \prod_{j=2}^m  \Phi\left(k^{(1)}_j, \dots, k^{(T)}_j; n_j; F_v\left(v_{1}^{(j)} \right)-F_v\left(v_{0}^{(j)} \right), \dots, F_v\left(v_{T}^{(j)}\right)-F_v\left(v_{T-1}^{(j)} \right) \right)  \\
& \times \frac{1}{n} \left( p_T \sum_{j=2}^m k^{(T)}_j  + p_{T-1} \sum_{j=2}^m k^{(T-1)}_j  + \dots + p_1 \sum_{j=2}^m k^{(1)}_j \right)\\
& \times   \alpha_1 \left( p_T  \left(1-F_v\left(v^{(1)}_T\right)\right)  + p_{T-1}  \left( F_v\left(v^{(1)}_T\right)- F_v\left(v^{(1)}_{T-1} \right) \right)  + \dots + p_1  \left( F_v\left(v^{(1)}_2\right)- F_v\left(v^{(1)}_1 \right) \right) \right),
\end{align*}
where we used Theorem \ref{Thm:BersteinSimplex} for $n_1 \to \infty$. Again using Theorem \ref{Thm:BersteinSimplex}, $m-1$ times for $j=2, \dots, m$ as $n_j \to \infty$, the normalized expected revenue becomes 
\begin{align}\label{eq:tempNonUniform2}
& p_T \sum_{j=1}^m \alpha_j \left(1- F_v\left(v_{T}^{(j)}\right)\right)   + \dots + p_1 \sum_{j=1}^m \alpha_j \left(F_v\left(v_{2}^{(j)}\right)-F_v\left(v_{1}^{(j)}\right)\right) = \sum_{t=T}^1 p_t \mathbf{1}^T A \left( F_v\left(\mathbf{v}_{t+1}\right)- F_v\left(\mathbf{v}_t\right) \right),
\end{align}
where we used the fact that $\mathbf{0} \le \mathbf{v}_1 \le \mathbf{v}_2 \dots \le \mathbf{v}_T \le \mathbf{1}$ which guarantees $F_v\left(\mathbf{v}_{t+1}\right) - F_v \left( \mathbf{v}_{t} \right) \ge 0$ for all $t=1, \dots, T$. At the end of this proof, we will show that Assumption \ref{assump:generalnetworkNonUniformValuations} guarantees $\mathbf{0} \le \mathbf{v}_1 \le \mathbf{v}_2 \dots \le \mathbf{v}_T \le \mathbf{1}$. Taking summation of Eq. \eqref{eq:tempNonUniform1} for $t=2,\dots, T$ leads to 
\begin{align*}
EA\left(1- F_v\left(\mathbf{v}_2 \right) \right) = (p_1- p_T) \mathbf{1}.
\end{align*}
Using this equation in Eq. \eqref{eq:tempNonUniform1} for $t=1$ leads to $\mathbf{v}_1= p_T \mathbf{1}$. Therefore, the normalized revenue can be written as 
\begin{align}\label{eq:tempNonUniform3}
& \sum_{t=T}^1 p_t \mathbf{1}^T A \left( F_v\left(\mathbf{v}_{t+1}\right)- F_v\left(\mathbf{v}_t\right) \right) \nonumber \\
& \stackrel{(1)}{=} \left( \sum_{t=T}^2 p_t \left(p_{t-1}- p_t\right) \left(\mathbf{1}^T E^{-1} \mathbf{1} \right) \right)+ p_1 \mathbf{1}^T A \left(F_v\left( \mathbf{v}_{2} \right) - F_v\left( \mathbf{v}_1 \right) \right) \nonumber \\
& = \left( \sum_{t=T}^2 p_t \frac{p_{t-1}- p_t}{\left( \frac{1}{\mathbf{1}^T E^{-1} \mathbf{1}} \right)} \right)+ p_1 \mathbf{1}^T A \left(\mathbf{1}- \left( EA\right)^{-1} \mathbf{1} (p_1- p_T) - F_v(p_T) \mathbf{1} \right) \nonumber \\
& = \left( \sum_{t=T}^2 p_t \frac{p_{t-1}- p_t}{\left( \frac{1}{\mathbf{1}^T E^{-1} \mathbf{1}} \right)} \right)+ p_1 \frac{\left(\frac{1}{\mathbf{1}^T E^{-1} \mathbf{1}} \right)+p_{T}- \left( \frac{1}{\mathbf{1}^T E^{-1} \mathbf{1}} \right) F_v(p_T)-p_1}{ \left( \frac{1}{\mathbf{1}^T E^{-1} \mathbf{1}} \right)},
\end{align}
where (1) follows from Eq. \eqref{eq:tempNonUniform1}.
Therefore, the monopolist's problem is to choose $\mathbf{p}=(p_T, \dots, p_1)$ that maximizes
\begin{align*}
& h(\mathbf{p})= \left( \sum_{t=T}^2 p_t \frac{p_{t-1}- p_t}{\left( \frac{1}{\mathbf{1}^T E^{-1} \mathbf{1}} \right)} \right)+ p_1 \frac{\left(\frac{1}{\mathbf{1}^T E^{-1} \mathbf{1}} \right)+p_{T}- \left( \frac{1}{\mathbf{1}^T E^{-1} \mathbf{1}} \right) F_v(p_T)-p_1}{ \left( \frac{1}{\mathbf{1}^T E^{-1} \mathbf{1}} \right)}.
\end{align*}
The first order conditions, results in  
\begin{align*}
\frac{\partial h }{\partial p_T}&  = - 2p_{T} + p_{T-1} + p_1- \left( \frac{1}{\mathbf{1}^T E^{-1} \mathbf{1}} \right) p_1 f_v(p_T) = 0, \\
\frac{\partial h }{\partial p_t}  & = p_{t+1}- 2p_{t}+ p_{t-1}=0, \quad t= T-1, \dots, 2, \\
\frac{\partial h }{\partial p_1} & = p_2 + \left( \frac{1}{\mathbf{1}^T E^{-1} \mathbf{1}} \right)+ p_T- \left( \frac{1}{\mathbf{1}^T E^{-1} \mathbf{1}} \right) F_v(p_T) - 2 p_1 =0.
\end{align*}
We will first find the solution of this set of equations and then show that with Assumption \ref{assump:generalnetworkNonUniformValuations}, the solution of first order conditions maximizes the revenue. Starting from the last equation, we obtain 
\begin{align*}
p_2= 2p_1 - \left( \left( \frac{1}{\mathbf{1}^T E^{-1} \mathbf{1}} \right)+ p_T - \left( \frac{1}{\mathbf{1}^T E^{-1} \mathbf{1}} \right) F_v(p_T) \right).
\end{align*}
 Plugging this into the equation corresponding to $\frac{\partial h }{\partial p_2}=0 $, leads to 
\begin{align*} 
 p_3= 3 p_1 - 2\left( \left( \frac{1}{\mathbf{1}^T E^{-1} \mathbf{1}} \right)+ p_T - \left( \frac{1}{\mathbf{1}^T E^{-1} \mathbf{1}} \right) F_v(p_T) \right).
\end{align*} 
 Repeating this argument leads to 
\begin{align}\label{eq:SolvingRecursion}
p_t = t p_1 - (t-1) \left( \left( \frac{1}{\mathbf{1}^T E^{-1} \mathbf{1}} \right)+ p_T - \left( \frac{1}{\mathbf{1}^T E^{-1} \mathbf{1}} \right) F_v(p_T) \right), \quad t=1, \dots, T.
\end{align}
Using Eq. \eqref{eq:SolvingRecursion} for $t=T$ and the equation corresponding to $\frac{\partial h }{\partial p_T}=0 $ yields
\begin{align}\label{eq:Proof:p1pT1}
p_T= T p_1 - (T-1) \left( \left( \frac{1}{\mathbf{1}^T E^{-1} \mathbf{1}} \right)+ p_T - \left( \frac{1}{\mathbf{1}^T E^{-1} \mathbf{1}} \right) F_v(p_T) \right).
\end{align}
Moreover, using Eq. \eqref{eq:SolvingRecursion} for $t=T-1$ in the equation corresponding to $\frac{\partial h }{\partial p_T}=0 $, gives 
\begin{align}\label{eq:Proof:p1pT2}
& -2 p_T + \left( (T-1) p_1 - (T-2) \left( \left( \frac{1}{\mathbf{1}^T E^{-1} \mathbf{1}} \right)+ p_T - \left( \frac{1}{\mathbf{1}^T E^{-1} \mathbf{1}} \right) F_v(p_T)\right) \right) \nonumber \\
& + p_1- \left( \frac{1}{\mathbf{1}^T E^{-1} \mathbf{1}} \right) p_1 f_v(p_T) =0.
\end{align} 
Combining Eq. \eqref{eq:Proof:p1pT1} and Eq. \eqref{eq:Proof:p1pT2}, we find the price in the first and last periods as 
\begin{align}\label{eq:SolvingRecursion2}
p_T& = \left( 1-F_v(p_T) \right) \left(\frac{1}{f_v(p_T)} - \frac{T-1}{T \left( \mathbf{1}^T E^{-1} \mathbf{1} \right)} \right), \nonumber \\
p_1& = \frac{1-F_v(p_T)}{f_v(p_T)}=(T-1) \frac{1-F_v(p_T)}{T \left( \mathbf{1}^T E^{-1} \mathbf{1} \right)} + p_T.
\end{align}
Invoking Eq. \eqref{eq:SolvingRecursion2} in Eq. \eqref{eq:SolvingRecursion} results in the optimal price sequence
\begin{align}\label{eq:priceseqtemp}
p_t & = (T-t) \frac{1-F_v(p_T)}{T \left( \mathbf{1}^T E^{-1} \mathbf{1} \right)} + p_T, \quad t=1, \dots, T.
\end{align}
Therefore, the price sequence starts from $p_T$ in the first periods and linearly increases. Also, note that the price $p_T$ given by the solution of Eq. \eqref{eq:SolvingRecursion2} is unique. We show this by establishing that 
\begin{align*}
1- \frac{1-F(x)}{f(x)} + \frac{T-1}{T}\frac{1}{\mathbf{1}^T E^{-1} \mathbf{1}} (1- F(x))
\end{align*}
is increasing. Taking derivative of this equation leads to 
\begin{align*}
1- \left( \frac{1-F(x)}{f(x)}\right)' - \frac{T-1}{T} \frac{f(x)}{\mathbf{1}^T E^{-1} \mathbf{1}} \ge 1- \left( \frac{1-F(x)}{f(x)}\right)' - \frac{f(x)}{\mathbf{1}^T E^{-1} \mathbf{1}} \ge 0,
\end{align*}
where we used Assumption \ref{assump:generalnetworkNonUniformValuations} in the last inequality.

We next show that the first order condition gives the optimal price sequence. The Hessian of $h(\cdot)$ is given by a symmetric matrix $M$ where $M_{ii}=-2$ for $i=1, \dots, T-1$, $M_{TT}= -2 - \left(1/\mathbf{1}^T E^{-1} \mathbf{1} \right) p_1 f_v'(p_T)$, $M_{i, i+1}=1$, for $i=1, \dots, T-1$, and $M_{1T}= 1- \left(1/\mathbf{1}^T E^{-1} \mathbf{1} \right) f_v(p_T)$. We next show that the Hessian is negative semidefinite. For any $\mathbf{x}\in \mathbb{R}^T$, we  show that $\mathbf{x}^T M \mathbf{x} \ge 0$. We have 
\begin{align}\label{eq:NSD}
\mathbf{x}^T M \mathbf{x}&= -2\sum_{i=1}^{T-1} x_i^2 - (2+  \left(\frac{1}{\mathbf{1}^T E^{-1} \mathbf{1}} \right) p_1 f_v'(p_T)) x_T^2 + 2 \sum_{i=1}^{T-1} x_i x_{i+1} + 2  \left( 1- \left(\frac{1}{\mathbf{1}^T E^{-1} \mathbf{1}} \right) f_v(p_T) \right)x_1 x_{T} \nonumber \\
& \stackrel{(1)}{\le}  -x_1^2 - x_T^2\left( 1+ \frac{p_1 f'(p_T)}{\mathbf{1}^T E^{-1}\mathbf{1}}\right) + 2  \left( 1- \left(\frac{1}{\mathbf{1}^T E^{-1} \mathbf{1}} \right) f_v(p_T) \right)x_1 x_{T}   \nonumber \\
& \stackrel{(2)}{=} -x_1^2 - x_T^2\left( 1+ \frac{\left( 1- F(p_T) \right) f'(p_T)}{f(p_T) \mathbf{1}^T E^{-1}\mathbf{1}}\right) + 2  \left( 1- \left(\frac{1}{\mathbf{1}^T E^{-1} \mathbf{1}} \right) f_v(p_T) \right)x_1 x_{T},
\end{align}
where (1) follows from $a^2+ b^2 \ge 2 ab$ and (2) follows from Eq. \eqref{eq:SolvingRecursion2}. In order to show Eq. \eqref{eq:NSD} is non-positive it suffices to show that the matrix 
\begin{align*}
N=\begin{pmatrix}
-1  & 1- \frac{f(p_T)}{\mathbf{1}^T E^{-1} \mathbf{1}} \\
1- \frac{f(p_T)}{\mathbf{1}^T E^{-1} \mathbf{1}} & -\left( 1+ \frac{\left( 1- F(p_T) \right) f'(p_T)}{f(p_T) \mathbf{1}^T E^{-1}\mathbf{1}}\right)
\end{pmatrix}
\end{align*}
is negative semidefinite. Since $N_{11} < 0$, in order to show $N$ is negative semidefinite, it suffices to show the determinant is non-negative. The determinant is 
\begin{align*}
\text{det}(N)= f(p_T) \left( 1- \frac{f(p_T)}{\mathbf{1}^T E^{-1} \mathbf{1}} - \left(\frac{1-F(p_T)}{f(p_T)}\right)' \right) \ge 0,
\end{align*}
where we used Assumption \ref{assump:generalnetworkNonUniformValuations} to obtain the inequality.

To complete the proof, it remains to verify $\mathbf{0} \le \mathbf{v}_1 \le \mathbf{v}_2 \dots \le \mathbf{v}_T \le \mathbf{1}$. This holds because using Eq. \eqref{eq:tempNonUniform1}, we obtain
\begin{align}\label{eq:interior1}
F\left( \mathbf{v}_{t+1} \right) - F\left( \mathbf{v}_t \right)  = A^{-1}E^{-1}\mathbf{1} (p_{t-1}- p_{t})\ge \mathbf{0} , \quad t=T, \dots, 2,
\end{align}
where the inequality holds because $p_{t-1} \ge p_t$ and from Assumption \ref{assump:generalnetworkNonUniformValuations}, we have $(EA)^{-1} \mathbf{1} \ge \mathbf{0}$. Therefore, we have $ \mathbf{v}_2 \le \mathbf{v}_3 \dots \le \mathbf{v}_T \le \mathbf{1}$. We will next show that $\mathbf{v}_2 \ge \mathbf{v}_1 \ge 0$. 
Inequality $\mathbf{v}_1 \ge 0$ evidently holds as we have $\mathbf{v}_1= p_T \mathbf{1} \ge 0$. Also, $\mathbf{v}_2 \ge \mathbf{v}_1$ is equivalent to 
\begin{align}\label{eq:interior2}
F_v\left( \mathbf{v}_1 \right)= F_v\left(p_T \right) \mathbf{1}  \le \mathbf{1}- \left( EA\right)^{-1} \mathbf{1} \left( p_1- p_T \right) = F_v \left( \mathbf{v}_2\right).
 \end{align} 
This holds because we have 
\begin{align*}
\left( EA\right)^{-1} \mathbf{1} \left( p_1- p_T \right) & \stackrel{(1)}{=} \left( EA\right)^{-1} \mathbf{1} \left( \frac{1-F_v(p_T)}{f_v(p_T)}- p_T \right)  \stackrel{(2)}{=} (1- F_v(p_T)) \left( EA\right)^{-1} \mathbf{1}\frac{T-1}{T\left(\mathbf{1}^T E^{-1} \mathbf{1}\right)}   \\
& \stackrel{(3)}{\le}  (1- F_v(p_T)) \mathbf{1} ,
\end{align*}
where (1) follows from $p_1= \frac{1-F_v(p_T)}{f_v(p_T)}$, (2) follows from $p_T= \left( 1-F_v(p_T) \right) \left(\frac{1}{f_v(p_T)} - \frac{T-1}{T \left( \mathbf{1}^T E^{-1} \mathbf{1} \right)} \right)$, and (3) follows from  Assumption \ref{assump:generalnetworkNonUniformValuations} and in particular $\left( EA\right)^{-1} \mathbf{1}\frac{1}{\mathbf{1}^T E^{-1} \mathbf{1}} \le \mathbf{1} \le \frac{T}{T-1}\mathbf{1}$. 
We next find the corresponding optimal normalized revenue. Plugging price sequence Eq. \eqref{eq:priceseqtemp} in Eq. \eqref{eq:tempNonUniform3} leads to the optimal normalized revenue
\begin{align*}
& \sum_{t=T}^2 \left((T-t) \frac{1-F_v(p_T)}{T \left( \mathbf{1}^T E^{-1} \mathbf{1} \right)} + p_T  \right) \frac{1-F_v(p_T)}{T \left( \mathbf{1}^T E^{-1} \mathbf{1} \right)} \left(\mathbf{1}^T E^{-1} \mathbf{1} \right) \\
& + \left( (T-1) \frac{1-F_v(p_T)}{T \left( \mathbf{1}^T E^{-1} \mathbf{1} \right)} + p_T \right) \frac{\left(\frac{1}{\mathbf{1}^T E^{-1} \mathbf{1}} \right)+p_{T}- \left( \frac{1}{\mathbf{1}^T E^{-1} \mathbf{1}} \right) F_v(p_T)-\left( (T-1) \frac{1-F_v(p_T)}{T \left( \mathbf{1}^T E^{-1} \mathbf{1} \right)} + p_T \right)}{ \left( \frac{1}{\mathbf{1}^T E^{-1} \mathbf{1}} \right)} \\
& =  \left(1- F_v(p_T) \right) \left(\frac{T-1}{2 T} \frac{1}{\mathbf{1}^T E^{-1} \mathbf{1}} \left(1- F_v(p_T) \right) + p_T \right). 
\end{align*}

\subsection*{Proof of Proposition \ref{Lem:PsiIncreasing}}
Using Theorem \ref{thm:NonUniformValuations}, the optimal normalized revenue is 
\begin{align}\label{eq:NonUniformValuationsRevenuePf}
p_T \left(1- F_v(p_T) \right)+ \frac{T-1}{2 T} \frac{1}{\mathbf{1}^T E^{-1} \mathbf{1}} \left(1- F_v(p_T) \right)^2,
\end{align}
where 
\begin{align}\label{eq:NonUniformValuationsRevenuePf2}
p_T= \left( 1-F_v(p_T) \right) \left(\frac{1}{f_v(p_T)} - \frac{T-1}{T } \frac{1}{\mathbf{1}^T E^{-1} \mathbf{1} } \right).
\end{align}
We next take the derivative of Eq. \ref{eq:NonUniformValuationsRevenuePf} with respect to $y= \frac{T-1}{T} \frac{1}{\mathbf{1}^T E^{-1} \mathbf{1}}$ and show that it is positive, implying that the optimal normalized revenue is increasing in both $\Psi(E)=\frac{1}{\mathbf{1}^T E^{-1} \mathbf{1}}$ and $T$. Taking derivative of Eq. \ref{eq:NonUniformValuationsRevenuePf} with respect to $y$ leads to 
\begin{align}\label{Eq:ProofTempDeriv1}
& \frac{d }{d y}\left( p_T \left(1- F_v(p_T) \right)+ \frac{y}{2}  \left(1- F_v(p_T) \right)^2 \right) \nonumber \\
& = \frac{d p_T }{d y} \left(1- F(p_T) \right) \left(-2 + y f(p_T)- \frac{(1-F(p_T)) f'(p_T)}{f^2(p_T)} \right) - \frac{1}{2} (1-F(p_T))^2.
\end{align}
Also, taking derivative of Eq. \eqref{eq:NonUniformValuationsRevenuePf2} with respect to $y$ yields
\begin{align} \label{Eq:ProofTempDeriv2}
\frac{d p_T }{d y}= -\frac{(1-F(p_T))}{2- f(p_T) y + \frac{1-F(p_T)f'(p_T)}{f^2(p_T)}}.
\end{align}
Plugging Eq. \eqref{Eq:ProofTempDeriv2} into Eq. \eqref{Eq:ProofTempDeriv1} gives 
\begin{align*}
& \frac{d }{d y}\left( p_T \left(1- F_v(p_T) \right)+ \frac{y}{2}  \left(1- F_v(p_T) \right)^2 \right) 
 = \frac{1}{2} \left( 1- F(p_T) \right)^2 \ge 0.
\end{align*}
Therefore, the optimal normalized revenue in increasing in $y$. Since $y= \frac{T}{T-1} \Psi(E)$ is increasing in both $T$ and $\Psi(E)$, the optimal normalized revenue is increasing in both $T$ and $\Psi(E)$.

The second derivative of the optimal normalized revenue with respect to $y$ is 
\begin{align*}
\frac{d}{d y} \frac{1}{2} \left( 1- F(p_T) \right)^2 = -(1- F(p_T)) f(p_T) \frac{d p_T}{d y}  
 = f(p_T) \frac{(1- F(p_T))^2}{2- f(p_T) y + \frac{(1-F(p_T)) f'(p_T)}{f^2(p_T)}}\ge 0,
\end{align*}
where we used Assumption \ref{assump:generalnetworkNonUniformValuations} in the inequality. Since $y$ is linear in $\Psi(E)$ (i.e., $y= \frac{T-1}{T} \Psi(E)$), the optimal normalized revenue is convex in $\Psi(E)$. 

The slope of the optimal price sequence is $(1- F(p_T)) \frac{\Psi(E)}{T}$. 
We next take the derivative of the slope with respect to $\Psi(E)$ and show it is non-negative. We have 
\begin{align*}
\frac{d}{d \Psi(E)} \left( (1- F(p_T)) \frac{\Psi(E)}{T} \right) & = \frac{1- F(p_T)}{T} - \frac{\Psi(E)}{T} f(p_T) \frac{d p_T}{d y} \frac{d y}{d \Psi(E)} \\
& = \frac{1- F(p_T)}{T} - \frac{\Psi(E)}{T} f(p_T) \frac{-(1- F(p_T))^2}{2- f(p_T) y + \frac{(1-F(p_T)) f'(p_T)}{f^2(p_T)}} \frac{T}{T-1} \\
& = \frac{1- F(p_T)}{T} + \frac{\Psi(E)}{T} f(p_T) \frac{(1- F(p_T))^2}{2- f(p_T) y + \frac{(1-F(p_T)) f'(p_T)}{f^2(p_T)}} \frac{T}{T-1} \ge 0,
\end{align*}
where we used Assumption \ref{assump:generalnetworkNonUniformValuations} in the inequality. 
Therefore, the slope of the optimal price sequence is increasing in $\Psi(E)$. 
\subsection*{Proof of Proposition \ref{Lem:Majorization}}
First note that given the in-degree sequence, we have $\mathbf{1}^T C \mathbf{1} = \mathbf{1}^T \tilde{C} \mathbf{1}= \sum_{i=1}^m d_i^{\text{(in)}}$. We next compare the terms $\mathbf{1}^T C^2 \mathbf{1}$ and $\mathbf{1}^T \tilde{C}^2 \mathbf{1}$. 
Since the imbalance sequence of $C$ majorizes the imbalance sequence of $\tilde{C}$, we have 
\begin{align}\label{eq:majorTemp1}
\sum_{i=1}^j  \tilde{c}_i^{\text{(out)}} \ge \sum_{i=1}^j  c_i^{\text{(out)}}, \quad \forall j=1, \dots, m.
\end{align}
Using Eq. \eqref{eq:majorTemp1} for $j=m$, yields
\begin{align*}
d_m^{\text{(in)}} \sum_{i=1}^m  (c_i^{\text{(out)}}- \tilde{c}_i^{\text{(out)}})  \le 0.
\end{align*}
Again, using Eq. \eqref{eq:majorTemp1} for $j=m-1$ and the previous inequality we obtain 
\begin{align*}
d_m^{\text{(in)}}(c_m^{\text{(out)}}- \tilde{c}_m^{\text{(out)}})+ d_{m-1}^{\text{(in)}} \sum_{i=1}^{m-1}  (c_i^{\text{(out)}}- \tilde{c}_i^{\text{(out)}}) \le d_m^{\text{(in)}} \sum_{i=1}^m  (c_i^{\text{(out)}}- \tilde{c}_i^{\text{(out)}}) \le 0,
\end{align*}
where we used $\sum_{i=1}^{m=1}(c_i^{\text{(out)}}- \tilde{c}_i^{\text{(out)}}) \le 0$ and $d_{m-1}^{\text{(in)}} \ge d_m^{\text{(in)}}$. Repeating this argument, leads to 
\begin{align*}
\sum_{i=1}^m  d_{i}^{\text{(in)}} (c_i^{\text{(out)}}- \tilde{c}_i^{\text{(out)}}) \le 0, 
\end{align*}
which completes the proof. 

\subsection*{Proof of Proposition \ref{thm:pricediscrimination}}\label{App:Proof:thm:pricediscrimination}
Similar to the proof of Theorem \ref{thm:NonUniformValuations} the critical thresholds defining the buyers' equilibrium satisfy
\begin{align}\label{eq:proof:thm:indiff2}
 EA (\mathbf{v}_{t+1} - \mathbf{v}_t) &= \mathbf{p}_{t-1} - \mathbf{p}_t, \quad t=T, \dots, 2, \nonumber \\
 \mathbf{v}_1- \mathbf{p}_1 + EA (\mathbf{1} - \mathbf{v}_2)&=0,
\end{align}
and the normalized revenue becomes 
\begin{align}\label{eq:proof:thm:indiff3}
& h(\mathbf{p}) \triangleq  \sum_{t=T}^{1} \mathbf{p}_t^T A \left(\mathbf{v}_{t+1} - \mathbf{v}_t \right) = \left( \sum_{t=T}^{2} \mathbf{p}_t^T E^{-1} \left(\mathbf{p}_{t-1} - \mathbf{p}_t \right) \right) + \mathbf{p}_1^T A \left(\mathbf{1}-\mathbf{p}_T \right) - \mathbf{p}_1^T E^{-1} \left(\mathbf{p}_1- \mathbf{p}_T \right).
\end{align}
Taking derivative of the revenue and multiplying by $E$ leads to 
\begin{align*}
E \frac{\partial h }{\partial p_T}&  =\mathbf{p}_{T-1} - 2 \mathbf{p}_{T} - EA \mathbf{p}_{1}+ \mathbf{p}_{1} = 0, \\
E \frac{\partial h }{\partial p_t}  & = \mathbf{p}_{t+1}- 2 \mathbf{p}_{t}+ \mathbf{p}_{t-1}=0, \quad t= T-1, \dots, 2, \\
E \frac{\partial h }{\partial p_1} & = \mathbf{p}_{2} + EA\mathbf{1} - EA \mathbf{p}_{T} - 2 \mathbf{p}_{1} + \mathbf{p}_{T} =0.
\end{align*}
Similar to the argument used in the proof of Theorem \ref{thm:NonUniformValuations}, the optimal price sequence becomes 
\begin{align}\label{eq:xtempsefr}
\mathbf{p}_t= (T-t) \mathbf{x} + \mathbf{p}_T, \quad t=T, \dots, 1.
\end{align}
Plugging this into the equation corresponding to $\frac{\partial h }{\partial p_1}=0$, we obtain 
\begin{align}\label{eq:xtemp}
\mathbf{x}= \frac{1}{T} EA \left(\mathbf{1} - \mathbf{p}_T \right).
\end{align}
Finally, using \eqref{eq:xtemp} in the equation corresponding to $\frac{\partial h }{\partial p_T}=0$, leads to $\mathbf{p}_T$ given by 
\begin{align*}
EA \mathbf{1} = \left(I + \left(I - \frac{T-1}{T} EA \right)^{-1} \right) EA \mathbf{p}_T.
\end{align*}
This equation simplifies to
\begin{align}\label{eq:pTtemp}
\mathbf{p}_T= \left( I + \left(I - \frac{T-1}{T} EA \right)^{-1} \right)^{-1} \mathbf{1}.
\end{align}
Plugging Eq. \eqref{eq:pTtemp} in Eq. \eqref{eq:xtemp} leads to 
\begin{align*}
\mathbf{x}=\frac{1}{T} EA \left(\mathbf{1} - \mathbf{p}_T \right) = \frac{1}{T-1} \left( \frac{2T}{T-1} \left( EA \right)^{-1} - I \right)^{-1},
\end{align*}
which gives the optimal price sequence given in the statement of Theorem \ref{thm:pricediscrimination}. 
We next show that the first order condition provides the optimal solution. To this end, we show that the Hessian is negative semidefinite. The Hessian is given by 
\begin{align*}
\begin{pmatrix}
    -2I & I & 0 & \dots  & I-EA \\
    I & -2I & I & \dots  & 0 \\
    \vdots & \vdots & \vdots & \ddots & \vdots \\
    0 & \dots & I  & -2I & I \\
    I-EA & 0 &  \dots & I &  -2I
\end{pmatrix}.
\end{align*}
 We need to show that for any $\mathbf{x}_1, \dots, \mathbf{x}_T \in \mathbb{R}^m$ we have 
\begin{align*}
-2 \sum_{i=1}^m \mathbf{x}_i^T I \mathbf{x}_i + \sum_{i\neq j} \mathbf{x}_i^T I \mathbf{x}_j - 2 \mathbf{x}_1^T EA \mathbf{x}_T \le 0.
\end{align*}
This follows by taking summation of the following inequalities:
\begin{align*}
 \frac{1}{2}\mathbf{x}_i^T I \mathbf{x}_i^T + \frac{1}{2}\mathbf{x}_{i+1}^T I \mathbf{x}_{i+1}^T & \ge \mathbf{x}_i I \mathbf{x}_{i+1}, \quad i=1, \dots, T-1, \\
 \frac{1}{2}\mathbf{x}_T^T (I-EA) \mathbf{x}_T^T + \frac{1}{2}\mathbf{x}_T^T (I-EA) \mathbf{x}_T^T &\ge \mathbf{x}_1 (I-EA) \mathbf{x}_{T}\\
 \frac{1}{2}\mathbf{x}_T^T EA \mathbf{x}_T^T + \frac{1}{2}\mathbf{x}_T^T EA \mathbf{x}_T^T & \ge 0,
\end{align*}
where the  first and last set of inequalities evidently hold and the second inequality follows from the assumption that $I- EA$ is positive semidefinite.
 Finally, note that since $I-EA$ is positive semidefinite, $I- EA\frac{T-1}{T}$ is invertible. 
 
 We will next show that the critical thresholds are interior. By Assumption we have $\left(I + \left(I- \frac{T-1}{T} EA \right)^{-1} \right)^{-1} \mathbf{1} \le \mathbf{1}$, showing  $\mathbf{p}_{t-1} \ge \mathbf{p}_t$ for $t=T, \dots, 2$. Therefore, using Eq. \eqref{eq:proof:thm:indiff2}, we obtain $\mathbf{1} \ge \mathbf{v}_T \ge \dots \ge \mathbf{v}_2$. We will next show $\mathbf{v}_2 \ge \mathbf{v}_1 \ge 0$. Since $\mathbf{v}_1 = \mathbf{p}_T$ and $\left(I + \left(I- \frac{T-1}{T} EA \right)^{-1} \right)^{-1} \mathbf{1} \ge \mathbf{0}$, we obtain $\mathbf{v}_1 \ge \mathbf{0}$. We next show that $\mathbf{v}_2 \ge \mathbf{v}_1$. Using Eq. \eqref{eq:proof:thm:indiff2}, we have 
 \begin{align*}
 \mathbf{v}_1 = \mathbf{p}_T \le \mathbf{1}- (EA)^{-1} \frac{T-1}{T} (EA) (\mathbf{1}- \mathbf{p}_T) \le \mathbf{1}- (EA)^{-1}  (\mathbf{p}_1- \mathbf{p}_T)= \mathbf{v}_2,
 \end{align*}
where the first inequality evidently holds and the second inequality follows from Eqs. \eqref{eq:xtempsefr} and \eqref{eq:xtemp}. Finally, using the Sherman-Morrison-Woodbury formula \citet[Section 0.7.4]{horn2012matrix} stated at the end of this proof, we can simplify Eq. \ref{eq:pTtemp} as follows
\begin{align*}
\left( I + \left(I - \frac{T-1}{T} EA \right)^{-1} \right)^{-1} \mathbf{1} = \frac{1}{2} \mathbf{1} - \frac{1}{2} \frac{T-1}{T} E \left(I - \frac{1}{2}\frac{T-1}{T} EA \right)^{-1} \frac{A \mathbf{1}}{2}.
\end{align*}
\begin{lemma}[Sherman-Morrison-Woodbury]
Suppose that a nonsingular matrix $A \in \mathbb{R}^{n \times n}$ has a known inverse $A^{-1}$ and consider $B=A+XR^{-1}Y$,in which $X$ is $n$-by-$r$, $Y$ is $r$-by-$n$, and $R$ is $r$-by-$r$ and nonsingular. If $B$ is nonsingular, then
\begin{align*}
B^{-1} = A^{-1} - A^{-1} X(R + Y A^{-1} X)^{-1}Y A^{-1}.
\end{align*}

\end{lemma} 

\subsection{Model with Utility from all Purchases}\label{sec:future}
In this subsection, we consider a variation of our model in which the utility of a buyer depends on the entire set of buyers who buy the product over $T$ periods. More specifically, the utility of buyer $i$ is
\begin{align*}
v^{(i)} - p_t +  \sum_{j \neq i} g_{ij} \mathbf{1}\{ j \text{ buys at } s \in \{1, \dots, T\} \},
\end{align*}
where $v^{(i)}$ is the valuation of buyer $i$ and $p_t$ is the posted price at round $t$. 
 We assume that customers are individually rational, i.e., if a customer purchases at round $t$, her utility considering only the purchases already happened, should be non-negative. This assumption is crucial as without this assumption given any price sequence all customers purchase the item at the round with minimum price. We next show via an example that in this setting the optimal price sequence can either be increasing or decreasing.
\begin{example}\label{example:NewModelTwoBuyersTwoPeriods}
\textup{
Let the valuations be uniform over $[0,1]$, $n=2$, and $g_{12}=g_{21}=g \le 1$. We first find the optimal non-decreasing price sequence and then find the optimal non-increasing price sequence. 
\begin{itemize}
\item Non-decreasing price sequence (i.e. $p_2 \le p_1$): the optimal prices are $p_2=p_1=1/2$ with revenue $\frac{1+g}{2}$. 
The seller's problem becomes 
\begin{align*}
\max_{p_2 \le p_1} ~ 2(1-p_2)^2 p_2 + 2 p_2 (1-p_2) \left(p_2+ p_1 \mathcal{P}_{[0,1]}\left( 1- \frac{p_1-g}{p_2}\right) \right),
\end{align*}
where the first term of the objective captures the case when both valuations are above $p_2$ which happens with probability $(1-p_2)^2$ with corresponding revenue $2p_2$. The second term captures the case when only one of the valuations is above $p_2$ whose probability is $2p_2(1-p_2)$. In this case, the monopolist obtains revenue $p_2$ from the buyer with valuation above $p_2$. The other buyer purchases if its valuation is above $p_1- g$ whose probability is $\mathcal{P}_{[0,1]}\left( 1- \frac{p_1-g}{p_2}\right)$ (i.e., $\mathbb{P}\left[ v\ge p_1-g \mid v \le p_2 \right]$) and its revenue is $p_1$. Also, note that if both valuations are below $p_2$, then none of the buyers purchases the product in the first period. Therefore, they generate no network effect and none of them will purchase in the second period with price $p_1$ which is higher than $p_2$. The solution to this optimization problem is $p_1=p_2=\frac{1}{2}$. 
\item Non-increasing price sequence (i.e., $p_2 \ge p_1$): note that in this case customers may postpone the purchase even if their utility is non-negative in order to pay the lower price offered in the subsequent period. If a buyer with valuation $v \ge p_2$ buys in the first period her expected utility becomes 
\begin{align}\label{eq:Extemptemp1}
v- p_2 + g (1- (p_1-g)),
\end{align}
where the term $(1- (p_1-g))$ is the probability of the other customer purchasing in one of the periods. On the other hand, if a buyer with valuation $v \ge p_2$ buys in the second period her expected utility becomes 
\begin{align}\label{eq:Extemptemp2}
v- p_1 + g (1- p_1),
\end{align}
where the term $(1- p_1)$ is the probability of the other customer purchasing in one of the periods. Comparing Eqs. \eqref{eq:Extemptemp1} and \eqref{eq:Extemptemp2}, a buyer with valuation $v \ge p_2$ purchases in the first period if and only if 
\begin{align*}
v- p_2 + g (1-p_1+g) \ge v-p_1+g(1-p_1),
\end{align*}
which leads to $g^2 \ge p_2-p_1$. If the prices does not satisfy this inequality, then both buyers purchase in the second period with price $p_1$ (i.e., lower price) and the optimal expected revenue becomes $2(1-p_1)p_1 \le \frac{1}{2}$. We next consider a price sequence that satisfies $g^2 \ge p_2-p_1$ and show the optimal expected revenue is always higher than $\frac{1}{2}$. In this case, the seller's problem becomes
\begin{align*}
\max_{p_2 \ge p_1} ~~ & 2 (1-p_2)^2 p_2 + 2p_2 (1-p_2) \left( \mathcal{P}_{[0,1]}\left( 1- \frac{p_1-g}{p_2}\right) p_1 + p_2 \right) + p_2^2 \left(2 p_1 \left( 1- \frac{p_1}{p_2} \right) \right) \\
\text{ s.t. } ~~ & g^2 \ge p_2-p_1,
\end{align*} 
where the first term of the objective captures the case when both valuations are above $p_2$ which happens with probability $(1-p_2)^2$ with corresponding revenue $2p_2$. The second term captures the case when only one of the valuations is above $p_2$ whose probability is $2p_2(1-p_2)$. In this case, the monopolist obtains revenue $p_2$ from the buyer with valuation above $p_2$. The other buyer purchases if its valuation is above $p_1- g$ whose probability is $\mathcal{P}_{[0,1]}\left( 1- \frac{p_1-g}{p_2}\right)$ (i.e., $\mathbb{P}\left[ v\ge p_1-g \mid v \le p_2 \right]$) and its revenue is $p_1$. The third terms of the objective captures the case when both valuations are below $p_2$. In this case, each player purchase in the second period with probability $1- \frac{p_1}{p_2}$ and obtains revenue $p_1$ from each purchase.   
The solution of this problem is:
\begin{enumerate}
\item If $g \ge \frac{1}{2}$, the optimal price sequence is $p_1=\frac{1}{2}$, $p_2=\frac{5}{8}$ with the optimal revenue $\frac{25}{32}$.
\item If $\frac{\sqrt{13}-1}{6} \le g < \frac{1}{2}$, the optimal price sequence is $p_1=g$, $p_2= \frac{1+g-g^2}{2}$, with the optimal revenue $\frac{\left( 1+g-g^2 \right)^2}{2}$. 
\item If $ \sqrt{2}-1 \le g < \frac{\sqrt{13}-1}{6}$, the optimal price sequence is $p_1=g$, $p_2= \frac{g(g+1)}{2}$, with the optimal revenue $2g\left( 1+g-2g^2-2g^3\right)$. 
\item If $g < \sqrt{2}-1$, , the optimal price sequence is $p_1 = \frac{1-g^2}{2}$, $p_2= \frac{1+g^2}{2}$, with the optimal revenue $\frac{1}{2} \left( 1+g+2g^2 - 2g^3 - 3g^4 + g^5\right)$. 
\end{enumerate}
\end{itemize}
Putting these two cases together the optimal revenue becomes the one plotted in \autoref{fig:NewModelTwoBuyersTwoPeriods}. 
}
\end{example}
\begin{figure}[t]
    \centering
        \includegraphics[width=4 in]{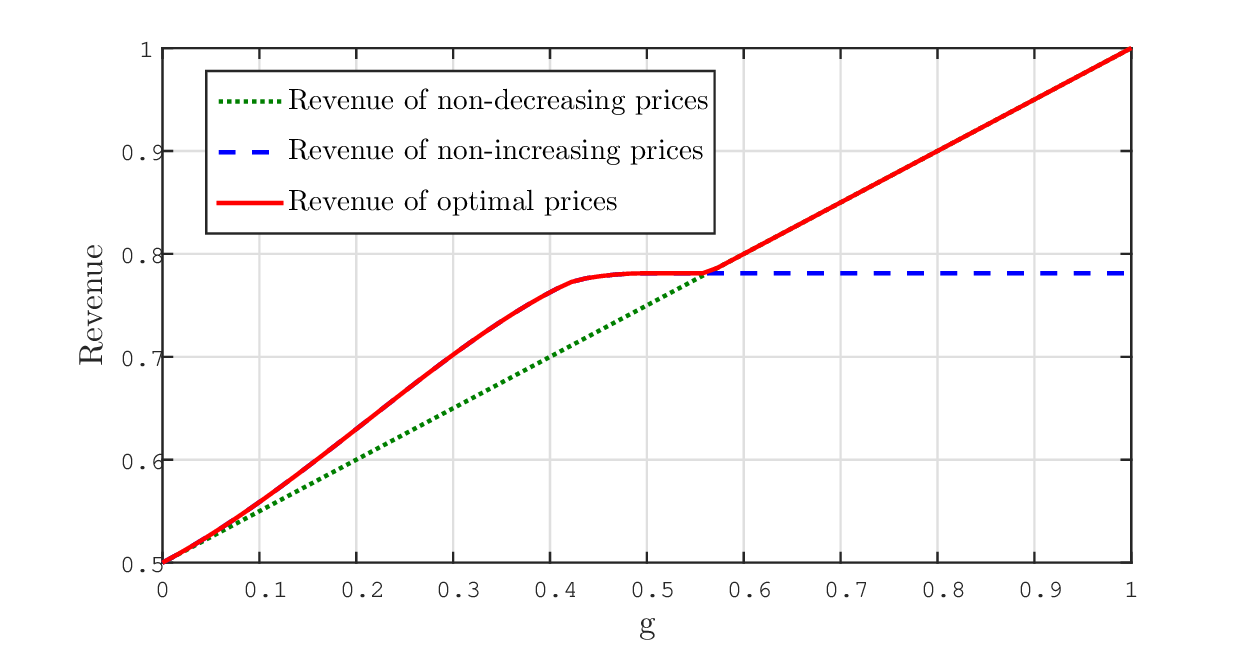}
        \caption{In a model with utility from all sales, the optimal price sequence can either be increasing or decreasing. This figure illustrates the optimal revenue for both increasing and decreasing price sequences for the setting described in Example \ref{example:NewModelTwoBuyersTwoPeriods}.}
      \label{fig:NewModelTwoBuyersTwoPeriods}
\end{figure}
In Example \ref{example:NewModelTwoBuyersTwoPeriods}, the optimal price sequence can be decreasing because a buyer who observes a lower price in the second period might be willing to purchase in the first period (with higher price) in order to incentivize the other buyer to purchase in the second period. This effect goes away once we consider a large number of buyers for which a similar argument to that of Proposition \ref{Pro:SingleCrrossing} shows that the optimal price sequence is non-decreasing. In the next proposition we characterize the optimal price sequence (which is non-decreasing). 
\begin{proposition}\label{Thm:Future}
Suppose the valuations are uniform and $||EA||_{\infty} < 1$. 
The optimal price sequence in the limit (as $n \to \infty$) is $p_t=\frac{1}{2}$, $t=T, \dots, 1$ with the optimal normalized revenue 
\begin{align*}
\frac{1}{4} \mathbf{1}^T A \left(I+ (EA)+ (EA)^2+ \dots + (EA)^{T-1} \right) \mathbf{1}. 
\end{align*}
\end{proposition}

Note that similar to Lemma \ref{Lem:DeltaSmall} for a weakly-tied block model $E=I+ \delta C$ and small enough $\delta$, the Assumption of Proposition \ref{Thm:Future} holds, i.e., $||EA||_{\infty} < 1$. 

Proposition \ref{Thm:Future} implies the following:
\begin{itemize}
\item The optimal revenue increases as the entries of the weighted network effects $EA$ increases. 
\item  Since $\rho(EA) \le ||EA||_{\infty}<1$, the limiting revenue as $T \to \infty$ becomes (see \citet[Chapter 5]{horn2012matrix})
\begin{align*}
\frac{1}{4} \mathbf{1}^T A (I- EA)^{-1} \mathbf{1}. 
\end{align*}

\item For a weakly-tied block model $E= I + \delta C$ and small enough $\delta$, we have $\rho(EA)<1$ and the first order Taylor series of $\frac{1}{4} \mathbf{1}^T A (I- EA)^{-1} \mathbf{1}$ leads to 
\begin{align*}
\frac{1}{4} \sum_{i=1}^m \frac{\alpha_i}{1- \alpha_i} + \delta \frac{1}{4}\sum_{i,j =1}^m \frac{\alpha_i}{1- \alpha_i} \frac{\alpha_j}{1- \alpha_j} C_{ij}.
\end{align*}
This shows that revenue becomes higher as the weighted summation of network externalities $C_{ij}$ increases, where the weights are given by $\frac{\alpha_i}{1- \alpha_i} \frac{\alpha_j}{1- \alpha_j}$. 
This establishes that revenue is higher when we have higher externality among blocks with higher sizes.
\end{itemize}
\subsection*{Proof of Proposition \ref{Thm:Future}}\label{App:Proof:Thm:Future}
The critical thresholds defining the buyers' equilibrium satisfy
\begin{align}\label{Pro:Future:Eq1}
\mathbf{v}_{t}- p_t \mathbf{1}+ EA (\mathbf{1}- \mathbf{v}_{t+1}) =\mathbf{0}, \quad t=1, \dots, T,
\end{align}
with the convention that $\mathbf{v}_{T+1}=1$. Eq. \eqref{Pro:Future:Eq1} leads to 
\begin{align}\label{Pro:Future:Eq2}
\mathbf{v}_{T+1}- \mathbf{v}_T & = (1- p_T)\mathbf{1} \nonumber \\
\mathbf{v}_{T}- \mathbf{v}_{T-1} & = p_T \mathbf{1}- p_{T-1} \mathbf{1} + (EA) \mathbf{1} (1- p_T) \nonumber \\
\mathbf{v}_{T-1}- \mathbf{v}_{T-2} & = p_{T-1}\mathbf{1} - p_{T-2} \mathbf{1} + (EA) \mathbf{1} (p_T- p_{T-1}) +  (EA)^2 \mathbf{1} (1- p_T) \nonumber \\
\mathbf{v}_{T+1-t}- \mathbf{v}_{T-t} & = \sum_{s=0}^t \left( (EA)^{t-s} \mathbf{1} \right) (p_{T+1-s}- p_{T-s}), \quad t=3, \dots, T-1,
\end{align}
with the convention $p_{T+1}=1$. The normalized revenue can be written as 
\begin{align}\label{Pro:Future:Eq3}
\sum_{t=1}^T p_t \mathbf{\alpha}^T (\mathbf{v}_{t+1}- \mathbf{v}_t)
\end{align}
which we need to maximize over non-decreasing price sequences, i.e., $p_T \le \dots \le p_2 \le p_1$. We show that the optimal solution is $p_t=\frac{1}{2}$, $t=1, \dots, T$. We establish this by using the sufficient KKT conditions for optimality (\citet[Proposition 3.3.2]{bertsekas1999nonlinear}). Following the notations used in \citet{bertsekas1999nonlinear}, we let 
\begin{align*}
& \mathbf{p}=(p_T, \dots, p_1), \\
& f_v(\mathbf{p})= -\sum_{t=1}^T p_t \mathbf{\alpha}^T (\mathbf{v}_{t+1}- \mathbf{v}_t),\\
& g_j(\mathbf{p})= p_{j+1}-p_j, \quad j=1, \dots, T-1, \\
& L(\mathbf{p}, \mathbf{\mu})= f_v(\mathbf{p}) + \sum_{j=1}^{T-1} \mu_j g_j(\mathbf{p}).
\end{align*}
With this notation, the optimization problem can be rewritten as 
\begin{align*}
& \min_{\mathbf{p}}  f_v(\mathbf{p}) \\
& \text{ s.t. } \mathbf{g}(\mathbf{p}) \le 0.
\end{align*}
Letting $\mathbf{p}^*= \frac{1}{2} \mathbf{1}$ and 
\begin{align*}
\mu^*_j= \frac{1}{2} \sum_{s=1}^{T-j} ((\mathbf{\alpha}^T(EA)^{s-1} \mathbf{1})- ((\mathbf{\alpha}^T(EA)^{T-s} \mathbf{1})), \quad j=1, \dots, T-1,
\end{align*}
we have 
\begin{align}
& \nabla_{\mathbf{p}} L(\mathbf{p}^*, \mathbf{\mu}^*)=0 \label{eq:KKT1}\\
& g_j(\mathbf{p}^*) \le 0, \quad j=1, \dots, T-1 \label{eq:KKT2} \\
& \mu_j^* (p^*_{j+1}-p^*_j) =0, \quad j=1, \dots, T-1 \label{eq:KKT3} \\
& \mu_j^*> 0,  \forall j, \text{ s.t. }p^*_{j+1}-p^*_j =0, \quad j=1, \dots, T-1 \label{eq:KKT4} \\
& \mathbf{y}^T  \nabla_{\mathbf{p} \mathbf{p}} L(\mathbf{p}^*, \mathbf{\mu}^*) \mathbf{y} > 0, \forall \mathbf{y} \neq 0 \text{ s.t. } \nabla g_j(\mathbf{p})^T \mathbf{y} =0, \forall j, \text{ s.t. }p^*_{j+1}-p^*_j =0 \label{eq:KKT5}. 
\end{align}
In particular, Eq. \eqref{eq:KKT1} and Eq. \eqref{eq:KKT2} are straightforward to verify, Eq. \eqref{eq:KKT3} holds as all inequalities are active and Eq. \eqref{eq:KKT4} holds because we have 
\begin{align*}
\mu^{*}_j & = \left(1- \mathbf{\alpha}^T (EA)^j \mathbf{1} \right) \sum_{s=0}^{T-1-j} \mathbf{\alpha}^T (EA)^s \mathbf{1}  \ge (1- ||(EA)^j||_{\infty})\sum_{s=0}^{T-1-j} \mathbf{\alpha}^T (EA)^s \mathbf{1} \\
& \ge (1- ||(EA)||_{\infty}^{j})\sum_{s=0}^{T-1-j} \mathbf{\alpha}^T (EA)^s \mathbf{1} > 0,
\end{align*}
where we used $||EA||_{\infty} <1$ in the last inequality. 
Finally, Eq. \eqref{eq:KKT5} holds because all $\mathbf{y} \in \mathbb{R}^{T-1}$ that satisfy the conditions are of the form $y \mathbf{1}$ for some non-zero $y \in \mathbb{R}$ for which 
\begin{align*} 
\mathbf{y}^T  \nabla_{\mathbf{p} \mathbf{p}} L(\mathbf{p}^*, \mathbf{\mu}^*) \mathbf{y} = y^2 \sum_{i, j=1}^{T} \nabla_{\mathbf{p} \mathbf{p}} L(\mathbf{p}^*, \mathbf{\mu}^*) = 2y^2 \sum_{t=1}^{T} \mathbf{\alpha}^T (EA)^{t-1} \mathbf{1} > 0.
\end{align*}
Therefore, using \citet[Proposition 3.3.2]{bertsekas1999nonlinear}, $p_t= \frac{1}{2}$, $t=1, \dots, T$ is the optimal price sequence which results in revenue 
\begin{align*}
\frac{1}{4} \mathbf{1}^T A \left(I+ (EA)+ (EA)^2+ \dots + (EA)^{T-1} \right) \mathbf{1}.
\end{align*} 
Finally, note that with price sequence $p_t=\frac{1}{2}$ for $t=1,\dots, T$, Eq. \eqref{Pro:Future:Eq2} results in $\mathbf{1} \ge \mathbf{v}_T \ge \dots \ge \mathbf{v}_1 \ge \mathbf{0}$, showing that the critical thresholds are in $[0,1]$. 


\bibliographystyle{plainnat}
\bibliography{References}
\end{document}